\documentclass[submission,copyright,creativecommons]{eptcs}
\providecommand{\event}{LSFA 2024} 

\usepackage{iftex}

\usepackage{amsmath,amssymb,soul}

\usepackage{bussproofs,caption,graphicx,float,tikz}
\usepackage{pdfpages}
\usepackage{algorithm}
\usepackage{mdframed}
\usepackage{mathpartir,enumerate}
\usepackage{algpseudocode,amsthm,tcolorbox}
\usepackage[colorinlistoftodos]{todonotes}
\usepackage{blindtext,multicol,stmaryrd}

\usepackage[english]{babel}

\usetikzlibrary{arrows.meta}
\usepackage{xy} 
\xyoption{all} 
\xyoption{knot}
\xyoption{arc}
\xyoption{import}
\xyoption{poly}

\usepackage{tikz}
\usetikzlibrary{arrows}
\usepackage{tkz-euclide}

\usepackage{lipsum}

\usepackage{array}

\usepackage{graphics} 

\newcommand{\aleq}{\approx_{\alpha}}
\newcommand{\caleq}{\approx_{\alpha,\theory{C}}}
\newcommand{\ealeq}{\approx_{\alpha,\theory{E}}}

\newcommand{\enarrow}[1]{\rightsquigarrow_{\theory{#1}}}

  {\gdef\scalefactor{#1}\begin{center}\proofSkipAmount \leavevmode}%
  {\scalebox{\scalefactor}{\DisplayProof}\proofSkipAmount \end{center} }

		\definecolor{bostonuniversityred}{rgb}{0.8, 0.0, 0.0}

\newcommand{\atomSet}{\mathbb{A}}

\newcommand{\pair}[1]{\langle #1\rangle}

\newcommand{\dom}[1]{\mathtt{dom}(#1)}

\newcommand{\nfpair}[1]{\langle #1\rangle_{nf}}
\newcommand{\pAction}[2]{#1\cdot #2}
\newcommand{\abs}[2]{[#1]#2}

\newcommand{\theory}[1]{\ensuremath{\mathsf{#1}}}

\newcommand{\context}[1]{\mathbb{C}}

\newcommand{\func}[1]{{\sf #1}}

\theoremstyle{plain}
\newtheorem{theorem}{Theorem}[section]
\newtheorem{corollary}[theorem]{Corollary}
\newtheorem{lemma}[theorem]{Lemma}
\newtheorem{proposition}[theorem]{Proposition}

\theoremstyle{definition}
\newtheorem{definition}[theorem]{Definition}
\newtheorem{example}{Example}[section]

\theoremstyle{remark}
\newtheorem{remark}[theorem]{Remark}

\ifpdf
  \usepackage{underscore}         
  \usepackage[T1]{fontenc}        
\else
  \usepackage{breakurl}           
\fi

\title{Nominal Equational Rewriting and Narrowing}

\author{
Mauricio Ayala-Rincón
\institute{University of Brasília, Brazil}
\and
Maribel Fernández
\institute{King's College London, UK}
\and
Daniele Nantes-Sobrinho
\institute{University of Brasília, Brazil}
\institute{Imperial College London, UK}
\and 
Daniella Santaguida\thanks{Author funded by Capes.}
\institute{University of Brasília, Brazil}
}
\def\titlerunning{Nominal Commutative Rewriting and Narrowing}
\def\authorrunning{M. Ayala-Rincón, M. Fernández, D. Nantes-Sobrinho \& D. Santaguida}
\def\copyrightholders{M. Ayala-Rincón, M. Fernández,\\ D. Nantes-Sobrinho \& D. Santaguida}
\begin{document}
\maketitle

\begin{abstract}
Narrowing is a well-known technique that adds to term rewriting mechanisms the required power to search for solutions to equational problems. Rewriting and narrowing are well-studied in first-order term languages, but several problems remain to be investigated when dealing with languages with binders using nominal techniques. Applications in programming languages and theorem proving require reasoning modulo $\alpha$-equivalence considering structural congruences generated by equational axioms, such as commutativity. 
This paper presents the first definitions of nominal rewriting and narrowing modulo an equational theory. We establish a property called nominal \theory{E}-coherence and demonstrate its role in identifying normal forms of nominal terms. Additionally, we prove the nominal \theory{E}-Lifting theorem, which ensures the correspondence between sequences of nominal equational rewriting steps and narrowing, crucial for developing a correct algorithm for nominal equational unification via nominal equational narrowing. We illustrate our results using the equational theory for commutativity.
\end{abstract}

\section{Introduction}

The nominal framework~\cite{VarBinding/GabbayP02} has emerged as a promising approach for dealing with languages involving binders such as lambda calculus and first-order logic.  In this framework,
equality coincides with the $\alpha$-equivalence relation, denoted as $\aleq$, and freshness constraints are integrated within the nominal reasoning rather than being relegated to the meta-language. For example, the expression  $a\# M$  (``$a$ is fresh for $M$'') indicates that if a name $a$ occurs in a term $M$, it must be abstracted by some binder, such as the $\lambda$ in the lambda calculus, or $\exists,\forall$-quantification in first-order logic, i.e., $a$  cannot occur free in $M$.

To enable reasoning within this framework, 
nominal  unification \cite{Matching/jcss/CalvesF10,NomUnification/UrbanPG04} was developed and  formalised in proof assistants such as Isabelle \cite{NomUnification/UrbanPG04}, PVS \cite{entcs:PVS/AnaCristina16} and Coq~\cite{A-C-AC/tcs/Ayala-RinconSFN19}. Nominal unification involves finding a substitution $\sigma$ that solves the problem $s\ _?{\overset{}{\approx}}_?\ t$, meaning $s\sigma\aleq t\sigma$, where $s$ and $t$ are nominal terms.  It is well-known that unification is fundamental for automated reasoning, serving as the foundation for resolution-based proof assistants, type inference, and numerous other applications. While these applications are anticipated to extend to nominal unification, substantial work is required to verify this.

To pursue applications of the nominal framework, extensions of nominal unification with equational theories have been investigated.  Initial efforts included integrating  the theories of \theory{A}ssociativity ($\approx_{\alpha,\theory{A}}$), \theory{C}ommutativity~($\approx_{\alpha,\theory{C}}$) and \theory{A}ssociativity-\theory{C}ommutativity ($\approx_{\alpha,\theory{AC}}$) to $\alpha$-equality~\cite{A-C-AC/tcs/Ayala-RinconSFN19}. Various algorithms for nominal unification modulo commutativity ($\theory{C}$-unification) and formalisations of their correctness in proof assistants PVS and Coq have been developed \cite{lopstr:NominalC-Unif/Washington17,OnSolvingNomFPequations/frocos/Washington17,FPconstraints/Ayala-RinconFN18,FormalisingNomC-unif/mscs/Gabriel21}. These development efforts reveal significant differences between first-order and nominal languages, such as the theory of \theory{C}-unification, which has nullary unification type if $\alpha$-equivalence is considered~\cite{lopstr:NominalC-Unif/Washington17}, contrasting with the finitary type of first-order \theory{C}-unification~\cite{Baader98}. 
 
Further investigations into nominal unification include exploring a \textit{letrec} constructor and extensions involving atom variables~\cite{NomUnif/Schmidt-Schauss22}. Another example is the development of an algorithm for nominal \theory{C}-matching~\cite{AFormalisationofNomCmatching/entcs/Ayala-RinconSFN19}, a special case of nominal \theory{C}-unification (dealing with problems $s\ _?{\overset{\theory{C}}{\approx}}_?\ t$ where the substitution $\sigma$ only applies in one side:  $s\sigma\caleq t$). Recently, a naive nominal extension of the  Stickel-Fages first-order  \theory{AC}-unification algorithm introduced cyclicity in solutions produced by translations of unification problems to Diophantine systems as reported in \cite{CICM:AyalaRinconFSKN23}, and this differs from the original (first-order) approach which has a terminating algorithm.

These developments underline the complexity of extending equational unification algorithms to the nominal framework, and new methods need to be proposed to obtain the desired extensions. An alternative approach to solving nominal unification problems modulo equational theories (i.e., nominal $\theory{E}$-unification problems), developed in \cite{NominalNarrowing16}, involves the use of {\em nominal narrowing}\footnote{Roughly, nominal narrowing is a generalisation of nominal rewriting by using nominal unification instead of nominal matching in its definition.}. This technique can be used when the equational theory $\theory{E}$ is presented by a convergent nominal rewriting system~\cite{NominalRewriting/FernandezG07,NomSyntax/fct/DominguezF19}. 
Different extensions are needed for rewriting modulo \theory{E} when such a presentation is impossible. However, nominal techniques modulo an equational theory \theory{E}, and in particular, nominal \theory{E}-rewriting, remain unexplored.

 This work represents the first step towards developing nominal \theory{E}-techniques, when using a convergent nominal rewrite system equivalent to the theory $\theory{E}$ is not possible.  
 In first-order term languages~\cite{JouannaudKK83:Incremental, Viola01, VariantNarrowing:EscobarMS09}, the standard technique is to split a set of identities \theory{T}  into a term rewriting system $\theory{R}$ and an equational part $\theory{E}$, so that $\theory{T}=\theory{R}{\cup}\theory{E}$, considering the rewriting relation generated by $\theory{R}$ on the equivalence classes of terms generated by $\theory{E}$. We propose extending this technique to nominal languages by adapting the notions of nominal rewriting~\cite{NominalRewriting/FernandezG07,Kikuchi020,Kikuchi22} and nominal narrowing~\cite{NominalNarrowing16}  to work modulo $\theory{E}$, incorporating the relation $\approx_{\alpha,\theory{E}}$.  These extensions result in the first definitions of nominal \theory{R/E}-rewriting (Definition \ref{def:RE-rewriting}) and \theory{R,E}-rewriting (Definition~\ref{def:rew-modC}) as well as \theory{E}-narrowing (Definition~\ref{def:narr-C}).
 
 Nominal \theory{R/E}-rewriting applies rules from \theory{R} in the equivalence class modulo $\approx_{\alpha, \theory{E}}$ of a nominal term $t$, while \theory{R,E}-rewriting uses nominal {\em $\theory{E}$-matching} to determine if a rule in \theory{R} applies to a nominal term, say $t$. This nominal term $t$ may have variables, thus the definition of the relations \theory{R/E} and \theory{R,E} also feature freshness conditions. We prove that it is possible to identify the normal form of a nominal term, say $t$, with respect to the relation \theory{R/E} (denoted $t\downarrow_\theory{R/E}$) to the normal form of the same term, but with respect to the relation \theory{R,E} (that is,  $t\downarrow_\theory{R,E}$), when the relation \theory{R,E} has a property called {\em nominal \theory{E}-coherence}, whose extension from a corresponding property in first-order term language~\cite{Jouannaud83:ConfluentandCoherent} is established here.

 
 Proving the correspondence between sequences of nominal \theory{R,E}-rewriting steps and \theory{E}-narrowing steps (Nominal \theory{E}-Lifting Theorem~\ref{teo:Clifitng}) is essential to develop an algorithm for nominal \theory{T}-unification via nominal \theory{E}-narrowing. Since the decidability of nominal \theory{T}-unification relies on the decidability of nominal \theory{E}-unification and \theory{E}-matching, and so far, the only equational theory for which a nominal unification algorithm exists is commutativity \theory{C}, a corollary of our developments is that the nominal \theory{C}-Lifting Theorem holds. Finally, due to the volume of extensions that were necessary to establish nominal \theory{E}-narrowing and rewriting, the final goal of using \theory{E}-narrowing as a sound and complete procedure for solving nominal \theory{T}-unification, remains ongoing work.




Summarising, our main contributions are:
\begin{enumerate}
 \item We extend the definitions and concepts regarding rewriting modulo $\theory{E}$ to the nominal framework. For instance, we have nominal versions of the relations $\to_{\theory{R,E}}$ and $\to_{\theory{R/E}}$ for rewriting, and $\rightsquigarrow_\theory{R,E}$ for nominal \theory{E}-narrowing.
 \item We prove technical auxiliary results relating $\to_{\theory{R,E}}$ and $\to_{\theory{R/E}}$. These required the establishment of the nominal \theory{E}-coherence property for ${\theory{R,E}}$.
 \item We  prove the nominal \theory{E}-Lifting Theorem (cf.\ Theorem~\ref{teo:Clifitng}) that establishes a correspondence between sequences of nominal \theory{E}-narrowing  $\rightsquigarrow_\theory{R,E}$ and nominal \theory{R,E}-rewriting $\to_{\theory{R,E}}$.
 \item Since $\theory{C}$ is the only equational theory for which a nominal unification algorithm exists, we illustrate our results using nominal $\theory{R,C}$-rewriting and narrowing.
   
\end{enumerate}

\paragraph*{Organisation.}  In \S\ref{sec:preliminaries} we present the background necessary to read the paper. Novel material starts in  \S\ref{sec:e-rewriting}, where we extend the notions of rewriting and narrowing modulo an equational theory \theory{E} to the nominal framework and provide some examples. In \S\ref{section:NomLiftingModuloC} we present the classical Lifting Theorem, extended to the nominal framework, taking into account an equational \theory{E} for which a nominal $\theory{E}$-unification algorithm exists. \S\ref{sect:future-work} concludes the paper.

\section{Preliminaries}\label{sec:preliminaries}

 While we assume the reader's familiarity with nominal techniques, we briefly recap some basic definitions. For more details, we refer to~\cite{NominalRewriting/FernandezG07}. In this (and the following) section(s), we will use $\equiv$ for syntactic equality, $=$ for definitions and $\approx_\alpha$ for $\alpha$-equality.

\paragraph{Syntax.}
Fix countable infinite pairwise disjoint sets of {\em atoms} \(\mathbb{A} = \{a,b,c,\ldots \}\) and {\em variables} \({\cal X} = \{X, Y, Z, \ldots \}\).
Atoms follow the \textit{atom convention}: atoms \(a, b, c,\ldots \)  over \(\atomSet\) represent different names.
Let \(\Sigma\) be a finite set of term-formers disjoint from \(\mathbb{A}\) and \({\cal X}\) such that for each \(f \in \Sigma\), a unique non-negative integer \(n\) (arity of \(f\)) is assigned.
A \textit{permutation} \(\pi\) is a bijection on \(\atomSet\) with finite domain, i.e., the set \(\dom{\pi} = \{a \in \atomSet \mid \pi(a) \neq a \}\) is finite. The identity permutation is denoted $id$. The composition of permutations $\pi$ and $\pi'$ is denoted $\pi\circ \pi'$ and $\pi^{-1}$ denotes the inverse of the permutation $\pi$.   

{\em Nominal terms} are defined  inductively by the grammar:
 $$ s,t,u ::= a \mid \pAction{\pi}{X} \mid  \abs{a}{t} \mid f(t_1, \ldots, t_n),$$
where \(a\) is an {\em atom}, \(\pAction{\pi}{X}\) is a (moderated/suspended) variable, \(\abs{a}{t}\) is the {\em abstraction} of \(a\) in the term \(t\), and \(f(t_1, \ldots, t_n)\) is a {\em function application} with \(f \in \Sigma\) and \(f:n\).  We abbreviate $id\cdot X$ as $X$. A term is \textit{ground} if it does not contain (moderated) variables. 
A \emph{position} $\context{C}$ is defined as a pair $(s, \_)$ of a term and a distinguished variable $\_ \in \mathcal{X}$ that occurs exactly once in $s$. We write $\context{C}[s']$ for $\context{C}[\_ \mapsto s']$ and if $s\equiv \context{C}[s']$, we say that $s'$ is a subterm of $s$ with position $\context{C}$. The root position will be denoted by $\context{C} = [\_]$. 

\begin{remark}[Positions]
  Our definition of ‘position’ is equivalent to the standard notion of a point in the abstract syntax tree of a term, as defined, for example, in~\cite{Baader98}. It is more convenient for us to identify this with the corresponding ‘initial segment’ of a nominal term, in which the `hole' is a variable in ${\cal X}$; thus positions of a term can be expressed within our language.
\end{remark}


A \emph{permutation action} of $\pi$ on a term $t$ is defined by induction on the term structure as expected:
$$\pi\cdot a = \pi(a) \qquad \pi \cdot (\pi'\cdot X) = (\pi \circ \pi')\cdot X \qquad \pi\cdot [a]t = [\pi\cdot a](\pi\cdot t) \qquad \pi\cdot f(t_1,\ldots, t_n) = f(\pi\cdot t_1, \ldots, \pi\cdot t_n).$$

The \emph{difference set} of two permutations
$ds(\pi,\pi') := \{ n \; | \; \pi\cdot n \neq \pi'\cdot n \}$. So $ds(\pi,\pi')\# X$ represents the set of constraints $\{ n\# X \; | \; n \in ds(\pi,\pi') \}$. For example, if $\pi = (a \ b)(c \ d)$ and $\pi' = (c \ b)$, then $ds(\pi,\pi') = \{a,b,c,d\}$ since $\pi$ and $\pi'$ act differently in each atom: note that $\pi(a)=b$ and $\pi'(a)=a$. In addition,  $ds(\pi,\pi')\# X = \{a\#X,b\#X,c\#X,d\#X\}$. 


A \emph{substitution} $\theta$ is a mapping from a finite set of variables to terms. The \emph{substitution action} $t\theta$ is defined as follows:
$$a\theta = a \qquad (\pi\cdot X)\theta = \pi\cdot(X\theta) \qquad ([a]t)\theta = [a](t\theta) \qquad f(t_1,\ldots, t_n)\theta = f(t_1\theta,\ldots, t_n\theta).$$
The domain of a substitution $\theta$ is written as $\texttt{dom}(\theta)$, and the image is denoted as $\texttt{Im}(\theta)$. Therefore, if $X \not\in \texttt{dom}(\theta)$ then $X\theta = X$. Also, if we restrict the domain to a certain set $V\subseteq {\cal X}$ of variables,  we obtain the substitution $\theta|_V$,  the \emph{restriction of $\theta$ to $V$}.
The identity substitution is denoted {\tt Id}.  The composition of two substitutions $\theta_1$ and $\theta_2$ will be denoted by simple juxtaposition 
as $\theta_1\theta_2$ and it applies to a term as $t\theta_1\theta_2=(t\theta_1)\theta_2$.

\paragraph{Nominal Constraints, Judgements and Rewriting.} There are two kinds of constraints: $s\aleq t$ is an (alpha-)equality constraint and $a\#t$ is a freshness constraint which means that $a$ cannot occur unabstracted in $t$. {\em Primitive constraints} have the form $a\#X$ and $\nabla,\Delta$ denote finite sets of primitive constraints. We will use the abbreviation $a,b,c\# X$ to denote the set of freshness constraints $\{a\# X, b\# X, c\# X\}$. {\em Judgements} have the form $\Delta\vdash s\aleq t$ and $\Delta\vdash a\# t$ and are derived using the rules in Figure~\ref{fig:fresh-and-equalrelation}. 

Given a finite set of freshness constraints $\Delta$ and a substitution $\theta$, $\Delta\theta$ consists of the set of constraints $\{a\#X\theta \mid a\#X \in \Delta\}$ and $\nfpair{\Delta\theta}$ consists of the set of freshness constraints obtained after applying the rules from Figure~\ref{fig:fresh-and-equalrelation} in $\Delta\theta$, in a bottom-up manner. $\nfpair{\Delta\theta}$ is {\em consistent} when it does not contain constraints of the form $a\#a$.
A \emph{problem $Pr$} is a set of constraints, and we write $\Delta \vdash Pr$ when for all $P \in Pr$ there is a derivation proof using the rules in Figure~\ref{fig:fresh-and-equalrelation}, taking elements of the context $\Delta$ as assumptions.

    \begin{figure}[!t]
    \centering
    \rule{15.3cm}{0.01cm}
    \begin{multicols}{2}
    \begin{prooftree}
    \AxiomC{}
    \RightLabel{\footnotesize (\# atom)}
    \UnaryInfC{$\Delta \vdash a \# b$}
    \end{prooftree} 
    \begin{prooftree}
    \AxiomC{$\Delta \vdash a \# t_1 \; \cdots \; \Delta \vdash a \# t_n$}
    \RightLabel{\footnotesize (\# app)}
    \UnaryInfC{$\Delta \vdash a \# f(t_1, \cdots, t_n)$}
    \end{prooftree}
    \end{multicols}
    \begin{multicols}{3}
    \begin{prooftree}
    \AxiomC{}
    \RightLabel{\footnotesize (\# a[a])}
    \UnaryInfC{$\Delta \vdash a \# [a]t$}
    \end{prooftree}
    \begin{prooftree}
    \AxiomC{$\Delta \vdash a \# t$}
    \RightLabel{\footnotesize (\# a[b])}
    \UnaryInfC{$\Delta \vdash a \#[b]t$}
    \end{prooftree}
    \begin{prooftree}
    \AxiomC{$(\pi^{-1} \cdot a \# X) \in \Delta$}
    \RightLabel{\footnotesize (\# var)}
    \UnaryInfC{$\Delta \vdash a \# \pi \cdot X$}
    \end{prooftree}
    \end{multicols}

\begin{multicols}{2}
\begin{prooftree}
\AxiomC{}
\RightLabel{\footnotesize ($\approx_\alpha$ atom)}
\UnaryInfC{$\Delta \vdash a \ \approx_\alpha \ a$}
\end{prooftree}

\begin{prooftree}
\AxiomC{$\Delta \vdash s_1 \ \approx_\alpha \ t_1 \; \cdots \; \Delta \vdash s_n \ \approx_\alpha \ t_n$}
\RightLabel{\footnotesize ($\approx_\alpha$ app)}
\UnaryInfC{$\Delta \vdash f(s_1, \cdots, s_n) \ \approx_\alpha \ f(t_1, \cdots, t_n)$}
\end{prooftree}
\end{multicols}

\begin{multicols}{2}
\begin{prooftree}
\AxiomC{$\Delta \vdash s \ \approx_\alpha \ t$}
\RightLabel{\footnotesize ($\approx_\alpha$ [aa])}
\UnaryInfC{$\Delta \vdash [a]s \ \approx_\alpha \ [a]t$}
\end{prooftree}

\begin{prooftree}
\AxiomC{$\Delta \vdash s \ \approx_\alpha \ (a\ b)\cdot t \qquad \Delta \vdash a \# t$}
\RightLabel{\footnotesize ($\approx_\alpha$ [ab])}
\UnaryInfC{$\Delta \vdash [a]s \ \approx_\alpha \ [b]t$}
\end{prooftree}
\end{multicols}

\begin{prooftree}
\AxiomC{$ds(\pi,\pi')\# X \in \Delta$}
\RightLabel{\footnotesize ($\approx_\alpha$ var)}
\UnaryInfC{$\Delta \vdash \pi \cdot X \ \approx_\alpha \ \pi' \cdot X$}
\end{prooftree}
\rule{15.3cm}{0.01cm}

\caption{Rules for $\#$ and $\approx_{\alpha}$}\label{fig:fresh-and-equalrelation}
\end{figure}




\begin{example}
Let $\Sigma_{\lambda}=\{{\sf lam}, {\sf app}\}$ denote the signature whose function symbols have arities ${\sf lam}:1$ and ${\sf app}:2$. Let $Pr=\func{lam}[a]\func{app}(a,X)\aleq \func{lam}[b]\func{app}(b,(a \ c)\cdot X)$ be a problem and  $\Delta = \{a,b,c\#X\}$ be a context. We verify the derivability of $\Delta\vdash \func{lam}[a]\func{app}(a,X)\aleq \func{lam}[b]\func{app}(b,(a \ c)\cdot X)$:

{\small
\begin{prooftree}
    \AxiomC{}
    \RightLabel{\footnotesize ($\aleq$ atom)}
    \UnaryInfC{$\Delta \vdash a \aleq a$}
    \AxiomC{$a,b,c\#X\in\Delta$}
    \RightLabel{\footnotesize ($\aleq$ var)}
    \UnaryInfC{$\Delta \vdash X \aleq (a\ b)(a\ c)\cdot X$}
    \RightLabel{\footnotesize ($\aleq$ app)}
    \BinaryInfC{$\Delta \vdash \func{app}(a,X)\aleq \func{app}(a,(a\ b)(a\ c)\cdot X)$}
    \AxiomC{}
    \RightLabel{\footnotesize ($\#$ atom)}
    \UnaryInfC{$\Delta\vdash a\#b$}
    \AxiomC{$c\# X \in \Delta$}
    \RightLabel{\footnotesize ($\#$ var)}
    \UnaryInfC{$\Delta \vdash a\# (a\ c)\cdot X$}
    \RightLabel{\footnotesize ($\#$ app)}
    \BinaryInfC{$\Delta\vdash a\# \func{app}(b,(a\ c)\cdot X)$}
    \RightLabel{\footnotesize ($\aleq$ [ab])}
    \BinaryInfC{$\Delta\vdash [a]\func{app}(a,X) \aleq [b]\func{app}(b,(a\ c)\cdot X)$}
    \RightLabel{\footnotesize ($\aleq$ app)}
    \UnaryInfC{$\Delta\vdash \func{lam}[a]\func{app}(a,X) \aleq \func{lam}[b]\func{app}(b,(a\ c)\cdot X)$}
\end{prooftree}
}
\end{example}

A {\em term in context} $\Delta\vdash t$ expresses that the term $t$ has the freshness constraints imposed by $\Delta$. For example, $a\#X\vdash f(X,h(b))$ expresses that $a$ cannot occur fresh in instances of $X$. Nominal rewriting rules can be defined under freshness constraints, i.e., $\nabla\vdash l\to r$ denotes a nominal rewriting rule. We denote by \theory{R}, a finite set of nominal rewriting rules.
     %
%

The {\em nominal rewriting relation} $\to_{\theory{R}}$ is defined as in~\cite{NominalRewriting/FernandezG07}: 
\begin{mathpar} 
\inferrule{ s\equiv\context{C}[s'] \qquad \Delta \vdash \big(\nabla \theta,\qquad  s'\aleq \pi\cdot (l \theta),\qquad  \context{C}[\pi \cdot (r \theta)] \aleq t\big)}{\Delta\vdash s \to_\theory{R} t}
\end{mathpar}
for a substitution $\theta$, a subterm $s'$ of $s$, a position $\context{C}$ and a nominal rule $\nabla\vdash l\to r\in \theory{R}$. We will omit the subscript \theory{R} and write only $\Delta\vdash s \to t$ when there is no ambiguity. 

 
\paragraph{Equality modulo an equational theory \theory{E}.} A nominal {\em identity}  is a  pair in context $\nabla\vdash (l,r)$ of nominal terms $l$ and $r$ under a (possibly empty) freshness context $\nabla$.  We denote such identity as $\nabla \vdash l\approx r$.  A set \theory{E} of identities induces an {\em equational theory}, which we will also denote as \theory{E}.

The \emph{nominal algebra equality modulo \theory{E}}, denoted  $\Delta \vdash s\ealeq t$, is the least transitive reflexive symmetric relation such that for any $(\nabla \vdash l\approx r)\in \theory{E}$, position $\context{C}$, permutation $\pi$, substitution $\theta$, and fresh context $\Gamma$ (so if $a\# X \in \Gamma$ then $a$ is not mentioned in $\Delta, s, t$):
\begin{prooftree}
    \AxiomC{$\Delta, \Gamma \vdash \big(\nabla\theta, \qquad s\approx_\alpha \context{C}[\pi\cdot(l\theta)], \qquad \context{C}[\pi\cdot (r\theta)]\approx_\alpha t \big)$}
    \RightLabel{($Ax_\theory{E}$)}
    \UnaryInfC{$\Delta \vdash s\approx_{\alpha,\theory{E}} t$}
\end{prooftree}


\begin{remark}\label{rmk:eq}We can also define $\ealeq$ by extending the rules of Figure~\ref{fig:fresh-and-equalrelation} with the dedicated rules for the identities defining \theory{E}.
 For example, the identity expressing the commutativity of a function symbol $f^\theory{C}$ is $\theory{C}=\{\vdash f^\theory{C}(X,Y)\approx f^\theory{C}(Y,X)\}$. In this case,  we need to add the following rule: 
\begin{mathpar}
\inferrule*[RIGHT=\mbox{\footnotesize $(\caleq \text{C})$}]{\Delta \vdash s_0 \ \caleq \ t_i \qquad \Delta \vdash s_1 \ \caleq \ t_{1-i} \\ i=0, 1 \\ f^\theory{C}\in \Sigma^\theory{C}}{\Delta \vdash f^\theory{C}(s_0, s_1) \ \caleq \ f^\theory{C}(t_0, t_1)}
\end{mathpar}
where $\Sigma^\theory{C}$ denotes a signature of commutative function symbols. 
Rule $(\aleq \text{app})$ only applies when the function symbol $f$ is not commutative. In addition, we need to modify the rules in Figure~\ref{fig:fresh-and-equalrelation} to use $\caleq$ instead of $\aleq$.
\end{remark}

Note that if we define an equational theory \theory{E}
 using the rule ($Ax_\theory{E}$), the equational theory is a congruence relation, and $\vdash$ is compatible with substitutions by definition. However, this rule generates a lot of redundant derivations. To avoid this, we will use specific rules for each \theory{E}, as for the commutative rule in Remark~\ref{rmk:eq}. This choice comes with a cost,
 now we need to prove the compatibility of $\vdash$ by substitution.

\begin{definition}[\theory{E}-Compatibility of $\vdash$ by substitutions]\label{def:e-compatible}
    An equational theory \theory{E} is {\em compatible with $\vdash$ by substitutions} iff the following hold, whenever $\pair{\Delta\theta}_{nf}$ is consistent. 

    \begin{enumerate}
    \item If $\Delta \vdash a\# t$ then $\pair{\Delta\theta}_{nf} \vdash a\# (t\theta)$. 
    \item If $\Delta \vdash s \ealeq t$ then $\pair{\Delta\theta}_{nf} \vdash (s\theta) \ealeq (t\theta) $.
    \item If $\Delta \vdash Pr$ then $\pair{\Delta\theta}_{nf} \vdash Pr\theta$.
\end{enumerate}
\end{definition}

The next proposition guarantees the compatibility of judgements by substitutions when the theory \theory{C} for commutativity is considered. This proposition is technical and will be used in the correspondence of one-step narrowing to one-step rewriting in Lemma~\ref{lem:correctness}.



\begin{proposition}
\label{lem:nf-subst}
The equational theory \theory{C} is compatible with substitutions.
\end{proposition}

\begin{proof}
By induction on the derivation of $\Delta \vdash a \# t$ or $\Delta \vdash s \caleq t$, using the rules of Fig.~\ref{fig:fresh-and-equalrelation}
extended for $\caleq$-equality.
\end{proof}

\paragraph{Nominal \theory{C}-unification algorithm.}
We consider the rule-based algorithm for nominal $\theory{C}$-unification,  introduced  in~\cite{OnSolvingNomFPequations/frocos/Washington17} and defined by the rules presented in Figure~\ref{fig:freshness-and-equal-simp}.
The rules act on triples  ${\cal P}=(\Delta, \theta, Pr)$, where $\Delta$ is a freshness context, $\theta$ is a substitution and 
$Pr$ is a \theory{C}-problem, i.e., a set of freshness and $\caleq$-equality constraints.  
We will denote the triples by $\mathcal{P}, \mathcal{Q}, \mathcal{S},\cdots$. 


\begin{definition}{($\theory{C}$-solution)}\label{def:csolution}
A \emph{$\theory{C}$-solution} for a triple $\mathcal{P} = (\Delta, \delta, Pr) $ is a pair $(\Delta',\theta)$ where the following conditions are satisfied:
\begin{enumerate}
    \item $\Delta' \vdash \Delta\theta$;
    \item $\Delta' \vdash a\#t\theta$, if $a\#t \in Pr$;
    \item $\Delta' \vdash s\theta \caleq t\theta$, if $s \caleq t \in Pr$;
    \item there is a substitution $\theta'$ such that $\Delta' \vdash \delta\theta' \caleq \theta$.
\end{enumerate}
If there is no $(\Delta', \theta)$, then we say that the problem $\mathcal{P}$ is \emph{unsolvable}. Also $\mathcal{U}_{\theory{C}}(\mathcal{P})$ denotes the set of all $\theory{C}$-solutions of the triple $\mathcal{P}$.
\end{definition}
Let 
 $(\Delta_1, \theta_1)$ and $(\Delta_2, \theta_2)$ be solutions in $\mathcal{U}_{\theory{C}}(\mathcal{P})$. We say that $(\Delta_1, \theta_1)$ is \emph{more general than} $(\Delta_2, \theta_2)$, and denote it as $(\Delta_1,\theta_1) \leq_{\theory{C}} (\Delta_2, \theta_2)$, if there exists a substitution $\theta'$ such that $\Delta_2 \vdash X\theta_1\theta' \caleq X\theta_2$, for all $X\in \mathcal{X}$ and $\Delta_2 \vdash \Delta_1\theta'$. 
We write $\leq^V_\theory{C}$ for the restriction of $\leq_{\theory{C}}$ to a set $V$ of variables.
\begin{definition}{(Nominal $\theory{C}$-unification problem)}\label{def:nominalCunif}
A \emph{nominal $\theory{C}$-unification problem (in context)} is a pair 
 $(\nabla\vdash l)\;_?{\overset{\theory{C}}{\approx}}_?  \; (\Delta\vdash s)$.
The pair $(\Delta',\theta)$ is an $\theory{C}$-solution, or  $\theory{C}$-unifier, of $(\nabla\vdash l) \;_?{\overset{\theory{C}}{\approx}}_?  \; (\Delta\vdash s)$ iff
$(\Delta',\theta)$ is a $\theory{C}$-solution of the triple
$ \mathcal{P}= (\{\nabla,\Delta \}, \texttt{Id}, \{l\caleq s\})$, that is, 
conditions (1)-(4) of Definition~\ref{def:csolution} are satisfied.
$\mathcal{U}_{\theory{C}}(\nabla\vdash l,\Delta\vdash s)$ denotes the set of all $\theory{C}$-solutions of $(\nabla\vdash l) \;_?{\overset{\theory{C}}{\approx}}_?  \; (\Delta\vdash s)$.
If $\nabla$ and $\Delta$ are empty we write  simply $\mathcal{U}_{\theory{C}}(l,s)$. 
A subset $\mathcal{V}\in \mathcal{U}_{\theory{C}}(\mathcal{P})$ is said to be a \emph{complete set of $\theory{C}$-solutions of $\mathcal{P}$} if for all $(\Delta_1,\theta_1) \in \mathcal{U}_{\theory{C}}(\mathcal{P})$, there exists $(\Delta_2, \theta_2) \in \mathcal{V}$ such that $(\Delta_2,\theta_2) \leq_{\theory{C}} (\Delta_1, \theta_1)$.
\end{definition}



\begin{figure}[!t]
\centering
\rule{15.3cm}{0.01cm}
    $$    
    \begin{array}{l@{\hspace{1cm}}rll}
        {\text{\footnotesize ($\#$ ab)}} & (\Delta, \theta, Pr \uplus \{a \# b \}) & \Longrightarrow & (\Delta, \theta, Pr) \\
        {\text{\footnotesize ($\#$ app)}} & (\Delta, \theta, Pr \uplus  \{a \# f(t_1,\cdots, t_n)\}) & \Longrightarrow & (\Delta, \theta, Pr \cup \{a \# t_1,\cdots, a \# t_n \}) \; \\
        {\text{\footnotesize ($\#$ a[a])}} & (\Delta, \theta, Pr \uplus \{a \# [a] t \}) & \Longrightarrow & (\Delta, \theta, Pr) \\
        {\text{\footnotesize ($\#$ a[b])}} & (\Delta, \theta, Pr \uplus \{a \# [b] t \}) & \Longrightarrow & (\Delta, \theta, Pr \cup \{a \#t\}) \\
        {\text{\footnotesize ($\#$ var)}} & (\Delta, \theta, Pr \uplus \{a \# \pi \cdot X\}) & \Longrightarrow & (\{(\pi^{-1} \cdot a)\# X\}\cup \Delta, \theta, Pr) \\
    \end{array}
    $$   


\rule{15.3cm}{0.01cm}
{\small
    $$    
    \begin{array}{lrll}
        {\text{\footnotesize ($\caleq$ refl)}} & (\Delta, \theta, Pr \uplus \{s \caleq s \}) & \Longrightarrow & (\Delta, \theta, Pr)\\
        {\text{\footnotesize ($\caleq$ app)}} & (\Delta, \theta, Pr \uplus \{f(\overline{s})_n \caleq f(\overline{t})_n \}) & \Longrightarrow & (\Delta, \theta, Pr \cup \bigcup\{ s_i \caleq t_i \} )\\
        {\text{\footnotesize ($\caleq$ C)}} & (\Delta, \theta, Pr \uplus \{f^\theory{C} s \caleq f^\theory{C} t \}) & \Longrightarrow & (\Delta, \theta, Pr \cup \{ s \caleq v \}), {\text{\footnotesize where \; $s=(s_0,s_1)$}}\\
        &&& {\text{\footnotesize\;and\; $t=(t_0,t_1)$, $v=(t_i, t_{(1-i)} , i = 0,1$ }}\\
        {\text{\footnotesize ($\caleq$ [aa])}} & (\Delta, \theta, Pr \uplus \{[a]s \caleq [a]t \}) & \Longrightarrow & (\Delta, \theta, Pr \cup \{ s \caleq t \})\\
        {\text{\footnotesize ($\caleq$ [ab])}} & (\Delta, \theta, Pr \uplus \{[a]s \caleq [b]t \}) & \Longrightarrow & (\Delta, \theta, Pr \cup \{ s \caleq (a\ b)\cdot t, a \# t \})\\
        {\text{\footnotesize ($\caleq$ inst)}} & (\Delta, \theta, Pr \uplus \{ \pi \cdot X \caleq t \}) & \Longrightarrow & (\Delta, \theta', Pr[X\mapsto\pi^{-1}\cdot t] \cup \underset{\overset{Y \in dom(\theta'),}{a \# Y \in \Delta}}{\bigcup} \{a \# Y\theta' \}),\\
        & && {\text{\footnotesize let\; $\theta':= \theta[X \mapsto \pi^{-1} \cdot t]$,}}\\
        & && {\text{\footnotesize if \; $X \not\in Var(t)$} }\\
        {\text{\footnotesize ($\caleq$ inv)}} & (\Delta, \theta, Pr \uplus \{ \pi \cdot X \caleq \pi'\cdot X \}) & \Longrightarrow & (\Delta, \theta, Pr \cup \{ (\pi')^{-1} \circ \pi \cdot X \caleq X \})\;\\
        & & & {\text{\footnotesize if \; $\pi' \neq \texttt{Id}$}}\\
\end{array}
$$    
}
\rule{15.3cm}{0.01cm}
    \caption{Simplification rules for $\#$ and $\caleq$. $\uplus$ denotes disjoint union}
    \label{fig:freshness-and-equal-simp}
\end{figure}


The following example illustrates the use of the nominal \theory{C}-unification algorithm to solve a nominal \theory{C}-unification problem.

\begin{example}\label{ex:c-unification} Let $\Sigma = \{h:1, f^\theory{C}:2, \oplus: 2\}$ be a signature, where $f^\theory{C}$ and $\oplus$ are commutative symbols, i.e., and $\theory{C} =\{ \ \vdash f^C(X,Y) \approx f^\theory{C}(Y,X), \ \vdash X\oplus Y \approx Y\oplus X\}$ be the axioms defining the theory. Consider the \theory{C}-unification problem
     $(\emptyset \vdash h(Y)) \; _?{\overset{\theory{C}}{\approx}}_? \; (\emptyset \vdash h(f^\theory{C}([b][a]X,X)))$  which has the associated triple $(\emptyset, \mathtt{Id}, \{h(Y) \; _?{\overset{\theory{C}}{\approx}}_? \; h(f^\theory{C}([b][a]X,X))\})$. By applying the rules from Figure~\ref{fig:freshness-and-equal-simp} we get the following:
{\small
$$
\begin{array}{l}
    (\emptyset, \mathtt{Id}, \{h(Y) \; _?{\overset{\theory{C}}{\approx}}_? \; h(f^\theory{C}([b][a]X,X))\})  
    \Longrightarrow_{(\caleq \text{\ app})}  (\emptyset, \mathtt{Id}, \{Y \; _?{\overset{\theory{C}}{\approx}}_? \; f^\theory{C}([b][a]X,X)\}) \\
    \qquad \qquad \qquad \Longrightarrow_{(\caleq \text{\ inst})}  (\emptyset, \theta_0 = [Y\mapsto f^\theory{C}([b][a]X,X)], \{f^\theory{C}([b][a]X,X) \; _?{\overset{\theory{C}}{\approx}}_? \; f^\theory{C}([b][a]X,X)\})  \\
    \qquad \qquad \qquad \Longrightarrow_{(\caleq \text{\ refl})}  (\emptyset, \theta_0 = [Y\mapsto f^\theory{C}([b][a]X,X)], \emptyset)\\
\end{array}
$$
}
Thus, we get the \theory{C}-solution $(\emptyset, \theta_0)$.

Now consider the \theory{C}-unification problem 
$(\emptyset\vdash f^\theory{C}([a][b]Z,Z)) \; _?{\overset{\theory{C}}{\approx}}_? \; (\emptyset\vdash  f^\theory{C}([b][a]X,X))$,
which has the associated triple $(\emptyset, \mathtt{Id}, \{f^\theory{C}([a][b]Z,Z)) \; _?{\overset{\theory{C}}{\approx}}_? \; f^\theory{C}([b][a]X,X)\})$. 
Using the Nominal \theory{C}-unification algorithm we get the following:
{\small
$$
\begin{array}{l}
    (\emptyset, \mathtt{Id}, \{f^\theory{C}([a][b]Z,Z)) \; _?{\overset{\theory{C}}{\approx}}_? \; f^\theory{C}([b][a]X,X)\}) \quad \Longrightarrow_{(\caleq \text{\ C})} \\ 
    \qquad \qquad  \qquad \Longrightarrow_{(\caleq \text{\ C})}  (\emptyset, \mathtt{Id}, \{[a][b]Z \; _?{\overset{\theory{C}}{\approx}}_? \; [b][a]X, Z \; _?{\overset{\theory{C}}{\approx}}_? \; X\}) \\
    \qquad \qquad \qquad \Longrightarrow_{(\caleq \text{\ inst})}  (\emptyset, \theta_1 = [Z\mapsto X], \{[a][b]X \; _?{\overset{\theory{C}}{\approx}}_? \; [b][a]X, X \; _?{\overset{\theory{C}}{\approx}}_? \; X\})  \\
    \qquad \qquad \qquad \Longrightarrow_{(\caleq \text{\ refl})}  (\emptyset, \theta_1, \{[a][b]X \; _?{\overset{\theory{C}}{\approx}}_? \; [b][a]X \}) \\
    \qquad \qquad \qquad \Longrightarrow_{(\caleq \text{\ [ab]})}  (\emptyset, \theta_1, \{[b]X \; _?{\overset{\theory{C}}{\approx}}_? \; (a\ b)\cdot[a]X, a\# [a]X \}) \\
    \qquad \qquad \qquad \Longrightarrow_{(\# \text{\ a[a]})}  (\emptyset, \theta_1, \{[b]X \; _?{\overset{\theory{C}}{\approx}}_? \; [b](a\ b)\cdot X \}) \\
    \qquad \qquad \qquad \Longrightarrow_{(\caleq \text{\ [bb]})}  (\emptyset, \theta_1, \{X \; _?{\overset{\theory{C}}{\approx}}_? \; (a\ b)\cdot X \}) \qquad\qquad \text{(Fixed-point problem)} \\
\end{array}
$$
}
Observe that the first step uses the rule $(\caleq \text{\ C})$, which yields two branches, but here, we are interested in analysing only one branch.

The fixed-point problem has infinite solutions, for example:
\begin{itemize}
    \item $(\{a\#X,b\#X\}, \rho_1)$: instances $\rho_1$ of $X$ that do not contain free occurrences of $a$ or $b$. E.g. for $\rho_1=[X\mapsto g(e)]$, we have  $X\rho_1=g(e) \caleq g(e)= (a\ b) \cdot (X\rho_1) $.
     \item $(\emptyset, \rho_2 
= [X\mapsto a \oplus b])$: since $X\rho_2 = a\oplus b \caleq b \oplus a = (a\ b)\cdot X\rho_2$
    \item $(\emptyset, \rho_3 
= [X\mapsto (a \oplus b)\oplus(a \oplus b)])$: since
$X\rho_3 = (a \oplus b)\oplus(a \oplus b) \caleq (b \oplus a)\oplus(b \oplus a) = (a\ b)\cdot X\rho_3.$ 
\end{itemize}
\end{example}

\section{Nominal \theory{E}-rewriting and \theory{E}-narrowing.} \label{sec:e-rewriting}

In this section, we introduce our novel definitions of {\em equational nominal rewriting systems} (ENRS) and {\em nominal equational narrowing}, sometimes abbreviated to nominal \theory{E}-rewriting systems and nominal \theory{E}-narrowing.

\subsection{Nominal \theory{E}-rewriting}
An {\em equational nominal rewrite system} (ENRS)   is a set of (nominal) identities $\theory{T}$ that can be split into a set $\theory{R}$ of nominal rewrite rules and a set $\theory{E}$ of identities. Sometimes, we will denote this decomposition as 
$\theory{R}{\cup}\theory{E}$.
%

\begin{definition}[Nominal \theory{R/E}-rewriting]\label{def:RE-rewriting}
Let $\theory{T}=\theory{R}{\cup}\theory{E}$ be an ENRS. The relation $\to_\theory{R/E}$ is induced by the composition $\ealeq \circ \to_\theory{R}\circ \ealeq$. A nominal term-in-context $\Delta \vdash s$ reduces with   $\to_{\theory{R/E}}$, when a term in its $\theory{E}$-equivalence class reduces via $\to_{\theory{R}}$ as below:
\begin{center}
$\Delta \vdash (s \to_{\theory{R/E}} t)$ iff there exist $s',t'$ such that $\Delta \vdash (s \ealeq s' \to_\theory{R} t' \ealeq t)$.
\end{center}

If $\Delta \vdash s \to_\theory{R/E}^* t$ and $\Delta \vdash s \to_\theory{R/E}^* u$, then we say  that $\theory{R}$ is \emph{$\theory{E}$-confluent} when there exist terms $t',u'$ such that $\Delta \vdash t \to_\theory{R/E}^* t'$, $\Delta \vdash u \to_\theory{R/E}^* u'$ and $\Delta\vdash t'\ealeq u'$.
Also, $\theory{R}$ is said to be \emph{$\theory{E}$-terminating} if there is no infinite $\to_\theory{R/E}$ sequence.
 $\theory{R}$ is called \emph{$\theory{E}$-convergent} if it is $\theory{E}$-confluent and $\theory{E}$-terminating.
\end{definition}
The following example illustrates an ENRS for the set of identities that define the prenex normal form of a first-order formula. We consider the commutativity of the connectives $\wedge$ and $\vee$.
\begin{example}[Prenex normal form rules]\label{ex:prenex-rules}
    Consider the signature for the first-order logic $\Sigma = \{\forall, \exists, \lnot, \land,\linebreak
    \lor \}$, let $\theory{C} = \{\;\vdash P \lor Q \approx Q \lor P, \;\vdash P \land Q \approx Q \land P\}$ be the commutative theory. The prenex normal form rules can be specified by the following set $\theory{R}$ of nominal rewrite rules:
\[
\begin{array}{rl}
    a \# P & \vdash \; P \land \forall[a]Q \rightarrow \forall [a](P \land Q)\\
    a \# P & \vdash \; P \lor \forall[a] Q \rightarrow \forall[a] (P \lor Q)\\
    a \# P & \vdash \; P \land \exists[a]Q \rightarrow \exists[a](P \land Q)\\
    a \# P & \vdash \; P \lor \exists[a]Q \rightarrow \exists[a](P \lor Q)\\
    & \vdash \; \lnot (\exists[a]Q) \rightarrow \forall[a] \lnot Q\\
    & \vdash \; \lnot (\forall[a]Q) \rightarrow \exists[a] \lnot Q\\
\end{array}
\]

\end{example}

Note that in Definition~\ref{def:RE-rewriting}, the relation $\to_\theory{R/E}$ deals with $\alpha,\theory{E}$-congruence classes and they are always infinite due to the availability of names for $\alpha$-renaming. Although  the pure $\approx_\alpha$ relation is decidable,
when $\aleq$ is put together with an equational theory $\theory{E}$ which contains infinite congruence classes, the relation $\to_{\theory{R/E}}$ may not be decidable (as in standard first-order rewriting modulo $\theory{E}$).  We will define the nominal relation  $\to_{\theory{R,E}}$ that deals with nominal $\theory{E}$-matching instead of inspecting the whole $\alpha,\theory{E}$-congruence class of a term. 



\begin{definition}[Nominal \theory{R,E}-rewriting]\label{def:rew-modC} 
The \emph{one-step \theory{E}-rewrite relation} $\Delta\vdash s\rightarrow_{\theory{R,E}} t$ is the least relation such that for any $R = (\nabla\vdash l\rightarrow r)\in \theory{R}$, position $\context{C}$, term $s'$, permutation $\pi$, and substitution $\theta$,
    \begin{prooftree}
    \AxiomC{$s\equiv \context{C}[s']$}
    \AxiomC{$\Delta\vdash \big( \nabla\theta,\ s' \ealeq \pi\cdot (l\theta), \ \context{C}[\pi\cdot (r\theta)] \aleq t\big)$}
    \BinaryInfC{$\Delta\vdash s\rightarrow_{\theory{R,E}}\  t $}
    \end{prooftree} 
    
The \theory{E}-\emph{rewrite relation} $\Delta \vdash s \rightarrow_{\theory{R,E}}^* t$ is
the least relation that includes $\rightarrow_{\theory{R,E}}$ and is closed by reflexivity and transitivity of $\to_{\theory{R,E}}$, i.e., it satisfies:
\begin{enumerate}
    \item  for all $\Delta, s, s'$ we have $\Delta \vdash s \rightarrow_{\theory{R,E}}^* s'$ if $\Delta \vdash s \aleq s'$;
    \item  for all $\Delta, s, t, u$ we have that $\Delta \vdash s \rightarrow_{\theory{R,E}}^* t$ and $\Delta \vdash t \rightarrow_{\theory{R,E}}^* u$ implies $\Delta \vdash s \rightarrow_{\theory{R,E}}^* u$.
\end{enumerate}

If
$\Delta \vdash s \to_\theory{R,E}^* t$ and $\Delta \vdash s \to_\theory{R,E}^* u$, then we say that \emph{\theory{R,E} is \theory{E}-confluent} when there exist terms $t',u'$ such that $\Delta \vdash t \to_\theory{R,E}^* t'$, $\Delta \vdash u \to_\theory{R,E}^* u'$ and $\Delta\vdash t'\ealeq u'$.
\end{definition}
 A term $t$ is said to be in {\em \theory{R,E}-normal} form (\theory{R/E}-normal form) whenever one cannot apply another step of $\to_\theory{R,E}$ ($\to_\theory{R/E}$).

    
\begin{example}[Cont. Example~\ref{ex:prenex-rules}]
    This example illustrates the  one-step \theory{C}-rewrite: $$a\#P' \vdash S'\lor (\exists [a]Q' \lor P') \to_{\theory{R,C}} S'\lor (\exists [a](P'\lor Q'))$$
    with the rule $a\#P \vdash P\lor \exists[a]Q \to \exists[a](P\lor Q)$.
    In fact,
    \begin{itemize}
        \item $\Delta = \{a\#P'\}$ and $\nabla = \{a\#P\}$;
        \item $s = S'\lor (\exists [a]Q' \lor P') \equiv \context{C}[\exists [a]Q' \lor P'] \equiv \context{C}[s']$;
    \end{itemize}
    If we fix $\pi=id$ and $\theta = [P \mapsto P', Q \mapsto Q']$ we have:
    \begin{itemize}
        \item $\Delta = a\#P' \vdash a\#P' = (a\#P)[P \mapsto P', Q \mapsto Q'] = \nabla\theta$;
        \item  $s' = \exists [a]Q' \lor P' \caleq (P \lor \exists [a]Q)[P \mapsto P', Q \mapsto Q'] = l\theta = \pi\cdot (l\theta)$;
        \item $\context{C}[\pi\cdot(r\theta)] = \context{C}[r\theta] = \context{C}[(\exists [a](P\lor Q))[P \mapsto P', Q \mapsto Q']] = \context{C}[\exists [a](P'\lor Q')] = S' \lor (\exists [a](P'\lor Q')) \aleq t$
    \end{itemize}
    Thus, $a\#P' \vdash S'\lor (\exists [a]Q' \lor P') \to_{\theory{R,C}} S'\lor (\exists [a](P'\lor Q'))$.

    Since $\lor$ is a commutative symbol, we could reduce the initial term to three other possible terms because we have two occurrences of the disjunction. Thus, we can ``permute'' the subterms inside the rewriting modulo $\theory{C}$.
\end{example}

\begin{remark}\label{rmk:e-coherence-prop}
Following the approach by Jouannaud et al. ~\cite{JouannaudKK83:Incremental}, $\theory{E}$-confluence is a consequence of relating $\to_{\theory{R/E}}$ and $\to_{\theory{R,E}}$, which relies on a property called {\em $\theory{E}$-coherence} which will be extended here, to the nominal framework.




\end{remark}

\vspace{2mm}
\noindent


 \begin{definition}[Nominal \(\theory{E}\)-Coherence] \label{def:e-coherence}
The relation $\Delta\vdash \_ \to_{\theory{R,E}} \_ $ is called \emph{$\theory{E}$-coherent} iff for all $ t_1, t_2, t_3$ such that $\Delta\vdash t_1\approx_{\alpha,\theory{E}} t_2$ and $\Delta\vdash t_1 \to_{\theory{R,E}} t_3$, there exist $t_4, t_5, t_6$ such that $\Delta\vdash t_3 \to_{\theory{R,E}}^* t_4$, $t_2 \to_{\theory{R,E}} t_5 \to_{\theory{R,E}}^* t_6$ and $\Delta\vdash t_4 \ealeq t_6$, for some $\Delta$.
\end{definition}
\begin{center}
\boxed{\includegraphics[scale=.5]{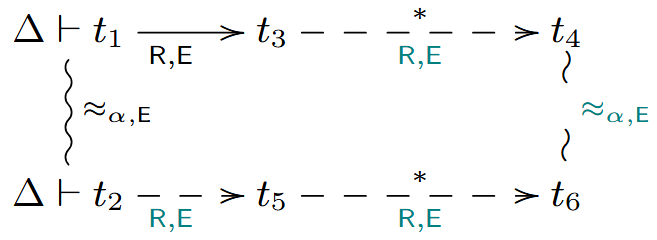}}
\end{center}

\vspace{2mm}
The diagram above illustrates nominal  \theory{E}-coherence: the dashed lines represent existentially quantified reductions.

\begin{definition}
An equational theory $\theory{E}$ is called {\em a first-order equational theory} iff $\theory{E}$ is defined via a set of first-order axioms, i.e., identities of the form $\emptyset\vdash l=r$, where $l,r$ are first-order terms. First-order terms do not contain atoms, abstractions and suspended permutations on variables. 
\end{definition}



\begin{theorem}\label{prop:nom-e-coherence}
Let \theory{E}
be a first-order theory, 
\theory{R} be a nominal rewrite system that is
\theory{E}-terminating
and \theory{R,E} be \theory{E}-confluent. 
Then the \theory{R,E}- and \theory{R/E}-normal forms of any term $t$ are \theory{E}-equal iff $\to_{\theory{R,E}}$ is \theory{E}-coherent.
\end{theorem}



In first-order rewriting, it is known that  $\theory{R,E}$-reducibility is decidable if  $\theory{E}$-matching is decidable.  Following Jouannaud et al.~\cite{JouannaudKK83:Incremental},   the existence of a finite and complete $\theory{E}$-unification 
algorithm is a sufficient condition for that decidability. 
However, solving nominal $\theory{E}$-unification problems has the additional complication of dealing with $\alpha$-equality, which significantly impacts obtaining finite and complete sets of nominal $\theory{E}$-unifiers.

\begin{remark}\label{rmk:cunif-notfin}
Nominal $\theory{C}$-unification is not finitary when one uses freshness constraints and substitutions for representing solutions~\cite{OnSolvingNomFPequations/frocos/Washington17}, but the type of problems that generate an infinite set of $\theory{C}$-unifiers are fixed-point equations $\pi\cdot X _?{\overset{\theory{C}}{\approx}}_?  X$. For example, the nominal \theory{C}-unification problem $(a \ b)\cdot X _?{\overset{\theory{C}}{\approx}}_?  X$  has solutions $[X\mapsto a\oplus b], [X\mapsto (a\oplus b)\oplus (a \oplus b)],\ldots$ (Example~ \ref{ex:c-unification}). However, these problems do not appear in nominal $\theory{C}$-matching, which is finitary~\cite{AFormalisationofNomCmatching/entcs/Ayala-RinconSFN19}. Thus, the relation $\to_{\theory{R,C}}$ is decidable. 
\end{remark}


\subsection{Nominal \theory{E}-narrowing}

Now we define the nominal narrowing relation modulo $\theory{E}$, extending previous works~\cite{NominalNarrowing16}.
\begin{definition}[Nominal \theory{E}-narrowing]\label{def:narr-C} 
The \emph{one-step \theory{E}-narrowing relation} $(\Delta\vdash s) \rightsquigarrow_{\theory{R,E}} (\Delta' \vdash t)$ is the least relation such that for any $(\nabla\vdash l\rightarrow r)\in \theory{R}$,  position $\context{C}$, term $s'$, permutation $\pi$, and substitution $\theta$, 
    \begin{prooftree}
    \AxiomC{$s\equiv \context{C}[s']$}
    \AxiomC{$\Delta'\vdash \big(\nabla\theta,~\Delta\theta,~ s'\theta \ealeq \pi\cdot (l\theta), \,(\context{C}[\pi\cdot r])\theta \aleq t\big)$}
    \RightLabel{.}
    \BinaryInfC{$(\Delta\vdash s) \rightsquigarrow_{\theory{R,E}}^\theta (\Delta' \vdash t) $}
    \end{prooftree}
where $(\Delta',\theta) \in {\cal U}_{\theory{E}}( \nabla\vdash l,  \Delta\vdash  s')$.  We will write only $(\Delta\vdash s) \rightsquigarrow_{\theory{R,E}} (\Delta' \vdash t)$, omitting the $\theta$, when it is clear in the context.

The \emph{nominal \theory{E}-narrowing relation} $(\Delta \vdash s) \rightsquigarrow_{\theory{R,E}}^* (\Delta' \vdash t)$ is 
the least relation that includes $\rightsquigarrow_{\theory{R,E}}$ and is closed by reflexivity and transitivity of $\rightsquigarrow_{\theory{R,E}}$, i.e., it  satisfies:
\begin{enumerate}
    \item for all $\Delta, s, s'$ we have $(\Delta \vdash s) \rightsquigarrow_{\theory{R,E}}^* (\Delta \vdash s')$ if $\Delta \vdash s \aleq s'$;
    \item  for all  $\Delta, \Delta', \Delta'',s, t$ and $u$: if  $(\Delta \vdash s) \rightsquigarrow_{\theory{R,E}}^* (\Delta' \vdash t)$ and $(\Delta' \vdash t) \rightsquigarrow_{\theory{R,E}}^* (\Delta'' \vdash u)$ then $(\Delta \vdash s) \rightsquigarrow_{\theory{R,E}}^* (\Delta'' \vdash u)$.
\end{enumerate}
\end{definition}
The permutation $\pi$ and substitution $\theta$ in the definition above are found by solving the nominal  $\theory{E}$-unification problem $(\nabla\vdash l) \; _?{\overset{\theory{E}}{\approx}}_? \; (\Delta\vdash  s')$.

\begin{remark} Note that decidability of $\rightsquigarrow_\theory{R,E}$ relies on the existence of an algorithm for nominal \theory{E}-unification. 
In this work, we will focus  on the theory \theory{C}, for which a nominal unification algorithm exists. 

\end{remark}

Since nominal $\theory{C}$-narrowing uses nominal $\theory{C}$-unification, which is not finitary when we use pairs $(\Delta', \theta)$ of freshness contexts and substitutions to represent solutions, following Remark~\ref{rmk:cunif-notfin}, we conclude that our nominal $\theory{C}$-narrowing trees are infinitely branching. The following example illustrates these infinite branches.
\begin{example}[Cont. Example~\ref{ex:c-unification}]\label{ex:narrow}
    Consider the signature $\Sigma = \{h:1, f^\theory{C}:2, \oplus: 2\}$, where $f^\theory{C}$ and $\oplus$ are commutative symbols.
    Let $R=\{ \ \vdash h(Y) \to Y ,\ \vdash f^\theory{C}([a][b]\cdot Z, Z) \to f^\theory{C}(h(Z),h(Z)) \}$ be a set of rewrite\footnote{$\vdash l\to r$ denotes $\emptyset\vdash l\to r$.} rules. 
    Let $\ \vdash h(f^\theory{C}([b][a]X,X))$ be a nominal term that we want to apply nominal \theory{C}-narrowing to.
    Observe that we can apply one step of narrowing, and then we obtain a branch that yields infinite branches due to the fixed-point equation (see Figure~\ref{fig:enter-label5}).

    \begin{figure}[!t]
    \centering
    \includegraphics[scale=.59]{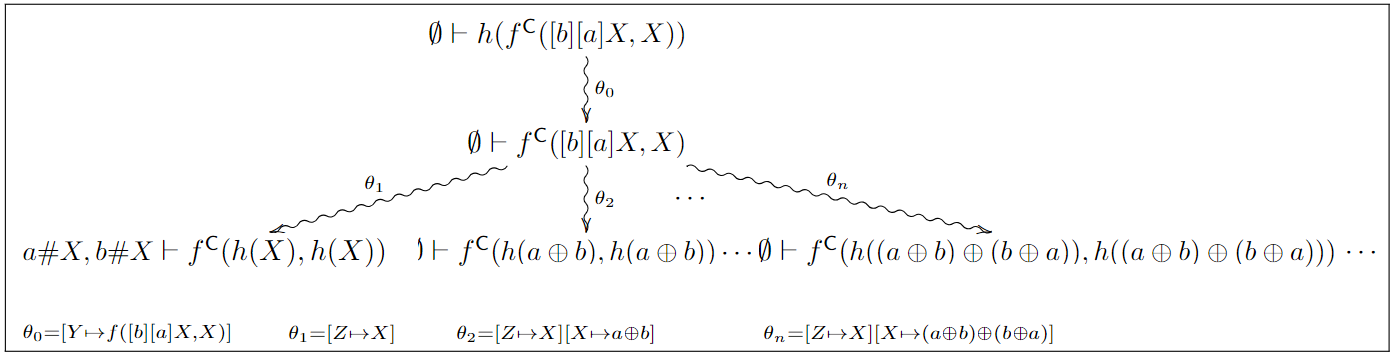}
    \caption{Infinitely branching tree}
    \label{fig:enter-label5}
    \end{figure}


    The first narrowing step is $\emptyset \vdash h(f^\theory{C}([b][a]X,X)) \enarrow{R,C} \emptyset \vdash f^\theory{C}([b][a]X,X)$, using the rule $\ \vdash h(Y) \to Y$.
    The substitution $\theta_0= [Y\mapsto f^\theory{C}([b][a]X,X)]$ was computed in Example~\ref{ex:c-unification} when solving the \theory{C}-unification problem  $(\emptyset\vdash h(Y)) \; _?{\overset{\theory{C}}{\approx}}_? \; (\emptyset\vdash  h(f^\theory{C}([b][a]X,X))).$


The other infinite narrowing steps are generated due to the fixed-point equation found in the process of solving the \theory{C}-unification problem 
$(\emptyset\vdash f^\theory{C}([a][b]Z,Z)) \; _?{\overset{\theory{C}}{\approx}}_? \; (\emptyset\vdash  f^\theory{C}([b][a]X,X))$, computed in Example~\ref{ex:c-unification}.
Composing the fixed-point solutions with $(\emptyset, \theta_1)$ that we had, we get the substitutions $\theta_1 = [Z\mapsto X]$, $\theta_2 = [Z\mapsto X][X\mapsto a \oplus b]$ and $\theta_3 = [Z\mapsto X][X\mapsto (a \oplus b)\oplus(a \oplus b)]$ of our narrowing steps in Figure~\ref{fig:enter-label5}.
 
\end{example}




 
 The following proposition shows that each nominal narrowing step corresponds to a nominal rewriting step, using the same substitution $\theta$.

\begin{proposition}\label{prop:narrow-to-rewriting}
Let \theory{E} be an equational theory for which a complete \theory{E}-unification algorithm exists. $(\Delta_0\vdash s_0) \rightsquigarrow_{\theory{R,E}}^{\theta} (\Delta_1 \vdash s_1) $ implies $\Delta_1\vdash (s_0 \theta ) \rightarrow_{\theory{R,E}}  s_1 $.
\end{proposition}

\begin{proof}
    Indeed, suppose we have $(\Delta_0\vdash s_0) \rightsquigarrow_{\theory{R,E}}^{\theta} (\Delta_1 \vdash s_1) $. The narrowing step guarantees that for a substitution $\theta$, some permutation $\pi$, and a rule $\nabla\vdash l\to r\in \theory{R}$, the following holds:  
    \begin{itemize}
        \item $s_0 \equiv \context{C}[s'_0]$ and $\Delta_1\vdash \big(\nabla\theta,~\Delta_0\theta,~ s'_0\theta \ealeq \pi\cdot (l\theta), \,(\context{C}[\pi\cdot r])\theta \aleq s_1\big)$. 
    \end{itemize}

   From the items above, it is easy to verify  the following:
   \begin{itemize} 
   \item $s_0\theta \equiv \context{C}\theta[s'_0\theta]$; and    $\Delta_1 \vdash (\nabla\theta', s_0'\theta \ealeq \pi\cdot (l\theta'), \context{C}\theta[\pi\cdot(r\theta')] \aleq s_1)$, 
   \end{itemize} 
   and  by the definition of rewrite modulo $\theory{E}$,  it implies that  $\Delta_1 \vdash s_0\theta \rightarrow_{\theory{R,E}} s_1$.
    We need to fix the substitution $\theta$ used in the narrowing step as $\theta'$, and the result follows.
\end{proof}


\section{Nominal Lifting Theorem modulo \theory{E}}\label{section:NomLiftingModuloC}

In this section, we assume $\theory{R}\cup \theory{E}$ an ENRS such that  $\theory{R}=\{\nabla_i\vdash l_i\to r_i\}$ is $\theory{E}$-convergent NRS, \theory{E} is compatible with $\vdash$ and substitutions and that there exists a complete \theory{E}-unification algorithm. We want to extend Proposition~\ref{prop:narrow-to-rewriting}  and establish correspondence between finite sequences of nominal $\theory{E}$-narrowing steps and sequences of nominal $\theory{E}$-rewriting steps. This result corresponds to the classical Lifting Theorem (\cite{Hullot80,JouannaudKK83:Incremental,NominalNarrowing16}) which will be extended to the nominal relations $\rightsquigarrow_{\theory{R,E}}$ and $\to_{\theory{R,E}}$. The Lifting Theorem relates narrowing steps to rewriting steps. It is fundamental to guarantee that one can use the narrowing relation to solve \theory{T}-unification problems when \theory{T} is a convergent equational theory. The extension to the \theory{R\cup E}-Lifting Theorem would allow us to solve nominal unification problems modulo $\theory{R}\cup \theory{E}$.

We start by defining a normalised substitution with respect to the relation $\to_{\theory{R,E}}$:
%

\begin{definition}[Normalised substitution w.r.t $\to_{\theory{R},\theory{E}}$]
A substitution $\theta$ is \emph{normalised in $\Delta$ with relation to $\to_{\theory{R,E}}$} if $\Delta \vdash X\theta$ is a $\theory{R,E}$-normal form  for every $X$. A substitution $\theta$ {\em satisfies the freshness context} $\Delta$ iff there exists a freshness context $\nabla$ such that $\nabla \vdash a\#X\theta$ for each $a\#X \in \Delta$. In this case, we say that $\theta$ satisfies $\Delta$ with $\nabla$. The minimal such $\nabla$ is $\nfpair{\Delta\theta}$.
\end{definition}


The following example illustrates the technique used in the proof of Lemma~\ref{lem:correctness}.
\begin{example}\label{ex:correct-example}
    Consider the rules $R_3: a\# P \vdash P\land\exists[a]Q\to \exists[a](P\land Q)$ and $R_6: \emptyset \vdash \neg(\forall[a]Q)\to \exists[a]\neg Q$.

    Let $(\Delta_0\vdash s_0) \rightsquigarrow^{\theta_0}_{R_6} (\Delta_1\vdash s_1) \rightsquigarrow^{\theta_1}_{R_3} (\Delta_2\vdash s_2)$ be a narrowing derivation, illustrated in Figure~\ref{fig:illustrate-example} such that:
    \begin{itemize}
        \item $\Delta_0 \equiv \emptyset$ and $s_0 \equiv P_1\land\neg(\forall[b]Q_1)$
        \item $\Delta_1 \equiv \{a\#Q_1\}$ and $s_1 \equiv P_1\land \exists[a](\neg (a\ b)\cdot Q_1)$ 
        \item $\Delta_2 \equiv \{a\# Q_1, a\#P_1\}$ and $s_2 \equiv \exists[a](P_1\land\neg(a\ b)\cdot Q_1)$ 
    \end{itemize}
    Let $\rho$ be a substitution that satisfies $\Delta_2$ with $\Delta$. Then, there exists a rewriting derivation
    $$\Delta\vdash s_0\rho_0 \to_\theory{R,C} s_1\rho_1 \to_\theory{R,C} s_2\rho$$
    where $\Delta\vdash \Delta_0\rho_0$, $\Delta\vdash \Delta_1\rho_1$ and $\rho_0=\theta_0\theta_1\rho$, $\rho_1 = \theta_1\rho$.

    \begin{figure}[!t]
    \begin{mdframed}
    \begin{center}
        \includegraphics[scale=.52]{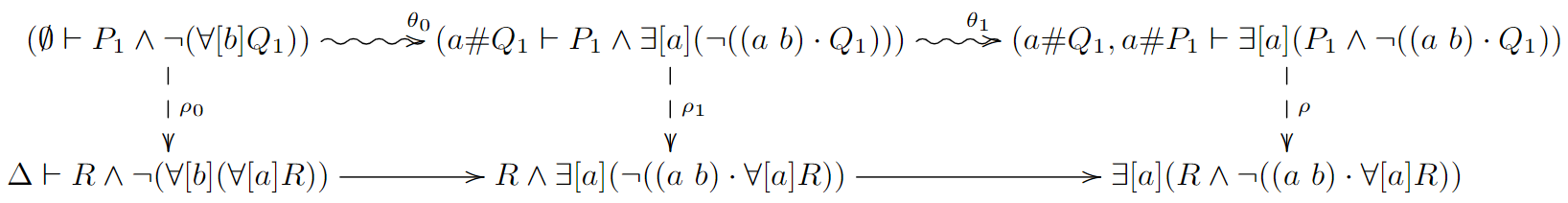}
    \end{center}
    \end{mdframed}
    \caption{Illustration of Example~\ref{ex:correct-example}}\label{fig:illustrate-example}
    \end{figure}
    Supposing that $\rho = [Q_1\mapsto \forall[a]R,P_1\mapsto R]$, and $\Delta =\{a\# R\}$, we have
    \begin{itemize}
        \item $\Delta\vdash \Delta_2\rho = \{a\#Q_1,a\#P_1\}\rho = \{a\#\forall[a]R,a\#R\} = \{a\#R\}$
        \item $\theta_0 = [Q \mapsto (a\ b)\cdot Q_1]$
         and  $\theta_1 = [P'\mapsto P_1, Q'\mapsto \neg((a\ b)\cdot Q_1)]$ 

        \item $\rho_1 = \theta_1\rho = [P' \mapsto R, Q'\mapsto \neg((a\ b)\cdot \forall[a]R), Q_1\mapsto \forall[a]R, P_1\mapsto R]$
        \item $\rho_0 = \theta_0\rho_1 = [Q\mapsto(a\ b)\cdot \forall[a]R, P' \mapsto R, Q'\mapsto \neg((a\ b)\cdot \forall[a]R), Q_1\mapsto \forall[a]R, P_1\mapsto R]$
        \item $\Delta\vdash \Delta_1\rho_1 = (a\#Q_1)\rho_1 = a\# \forall[a]R = \emptyset$ and  $\Delta\vdash \Delta_0\rho_0 = \emptyset$
    \end{itemize}
\end{example}


The next result shows that the rewriting step generated by the narrowing step
is preserved by application of substitution if the theory \theory{E} is compatible (Definition~\ref{def:e-compatible}), that is, $\vdash$ is closed by substitutions.
\begin{lemma}{($\rightsquigarrow_{\theory{R,E}}$ to $\to_{\theory{R,E}}$)}
\label{lem:correctness}
Let \theory{E} be compatible with $\vdash$ by substitutions and  $(\Delta_0 \vdash s_0) \rightsquigarrow_{\theory{R,E}}^\theta (\Delta_1 \vdash s_1) $. Then, for any substitution $\rho$ that satisfies $\Delta_1$ with $\Delta$,
the following holds
$$\Delta \vdash (s_0 \theta )\rho \rightarrow_{\theory{R,E}}  s_1 \rho$$
In particular, $\Delta$ will be $\nfpair{\Delta_1\rho}$.
\end{lemma}
\begin{proof}
From Proposition~\ref{prop:narrow-to-rewriting}:
$(\Delta_0\vdash s_0) \rightsquigarrow_{\theory{R,E}}^{\theta} (\Delta_1 \vdash s_1) $ implies $\Delta_1\vdash (s_0 \theta ) \rightarrow_{\theory{R,E}}  s_1 $.
By \theory{E}-compatiblity of derivability with substitutions $\Delta_1 \vdash s_0\theta \rightarrow_{\theory{R,E}} s_1$ gives: 

\begin{itemize}
    \item $(s_0\theta)\rho \equiv (\context{C}\theta[s'_0\theta])\rho = \context{C}\theta\rho[(s'_0\theta)\rho]$
    \item $\Delta_1 \vdash \nabla\theta$ implies  $\nfpair{\Delta_1\rho} \vdash \nabla\theta\rho$
    \item $\Delta_1 \vdash s_0'\theta \ealeq \pi\cdot (l\theta)$ implies $\nfpair{\Delta_1\rho} \vdash s_0'\theta\rho \ealeq (\pi\cdot(l\theta))\rho = \pi\cdot(l\theta\rho)$
    \item $\Delta_1 \vdash \context{C}\theta[\pi\cdot(r\theta)] \aleq s_1$ implies $\nfpair{\Delta_1\rho} \vdash \context{C}\theta\rho[\pi\cdot(r\theta\rho)] = (\context{C}\theta[\pi\cdot(r\theta)])\rho \aleq s_1\rho$
\end{itemize}
which implies that $\nfpair{\Delta_1\rho} \vdash (s_0\theta)\rho \rightarrow_{\theory{R,E}} s_1\rho$.
Note that we need $\rho$ satisfying $\Delta_1$ with $\Delta$ to guarantee that when we instantiate $\Delta_1$ we do not have any inconsistency with the freshness constraints in $\Delta_1$.
\end{proof}

The following result (correctness) states that a finite sequence of rewriting steps exists for each finite sequence of narrowing steps.


\begin{lemma}{($\rightsquigarrow_{\theory{R,E}}^*$ to $\to_{\theory{R,E}}^*$)}\label{theo:narrowtorewrite}
Let \theory{E} be compatible with $\vdash$ by substitutions and  $(\Delta_0\vdash s_0) \rightsquigarrow^*_{\theory{R,E}} (\Delta_n \vdash s_n) $ be a nominal \theory{E}-narrowing derivation. 
Let $\rho$ be a substitution satisfying $\Delta_n$ with  $\Delta$.
$$(\Delta_0\vdash s_0) \rightsquigarrow^{\theta_0}_{\theory{R,E}} (\Delta_1 \vdash s_1) \rightsquigarrow^{\theta_1}_{\theory{R,E}}  \ldots \rightsquigarrow^{\theta_{n-1}}_{\theory{R,E}} (\Delta_n \vdash s_n) $$
Then, there exists a nominal \theory{E}-rewriting derivation 
$$\Delta\vdash s_0\rho_0  \rightarrow_{\theory{R,E}} \ldots \rightarrow_{\theory{R,E}} s_i\rho_i \rightarrow_{\theory{R,E}}\ldots   \rightarrow_{\theory{R,E}} s_{n-1}\rho_{n-1}\rightarrow_{\theory{R,E}} s_n\rho$$ 
such that $\Delta\vdash \Delta_i\rho_{i}$ and $\rho_{i}=\theta_{i}\ldots \theta_{n-1}\rho$, for all $0\leq i < n$. In other words,
$\Delta\vdash (s_0 \theta )\rho \rightarrow^*_{\theory{R,E}}  s_n \rho $
where $\theta=\theta_0\theta_1\ldots \theta_{n-1}$.
\end{lemma}
\begin{proof}
By induction on the length $n\geq 1$ of the narrowing derivation $(\Delta_0\vdash s_0) \rightsquigarrow^{n}_{\theory{R,E}} (\Delta_n \vdash s_n)$, 
using the one-step result proved in Lemma~\ref{lem:correctness}. (We start the induction for $n=1$ because the case for $n=0$ holds trivially and gives no additional insight.) 
\begin{itemize}
    \item \textbf{Base Case:} For $n=1$, 
    we have $(\Delta_0\vdash s_0) \rightsquigarrow_{\theory{R,E}} (\Delta_1 \vdash s_1)$ and by Lemma~\ref{lem:correctness}, for any $\rho$ satisfying $\Delta_1$ with $\Delta$ we have $\Delta \vdash  (s_0\theta_0)\rho \rightarrow_{\theory{R,E}} s_1\rho$. Since $\Delta\vdash \Delta_1\rho$, and by the narrowing step $\Delta_1\vdash \Delta_0\theta_0$, we get $\Delta\vdash \Delta_0\theta_0\rho$. Taking $\rho_0 = \theta_0\rho$, we have the result
    $\Delta\vdash s_0\rho_0 \to_\theory{R,E} s_1\rho$
    such that $\Delta\vdash \Delta_0\rho_0$.
    \item \textbf{Induction Step:} Assume that the result holds for $n>1$. Then 
    $(\Delta_0 \vdash s_0) \rightsquigarrow^{n}_{\theory{R,E}} (\Delta_n \vdash s_n)$
    implies that there exists a rewriting derivation
    $\Delta \vdash s_0\rho_0 \rightarrow^{n}_{\theory{R,E}} s_n\rho$, for some $\rho$ satisfying $\Delta_n$ with $\Delta$
    and Figure~\ref{fig:narrrew} illustrates this setting.
    
    We want to show that the result follows for $n+1$. Consider the narrowing step 
    $$(\Delta_n \vdash s_n) \rightsquigarrow^{\theta_n}_\theory{R,E} (\Delta_{n+1} \vdash s_{n+1}).$$

   By Lemma~\ref{lem:correctness}, for any substitution, let's name it $\sigma$, that satisfies $\Delta_{n+1}$ with $\Delta$ {\bf (H1)}
    we have
    \begin{equation}\label{eq:narrow}
        \Delta \vdash (s_n\theta_n)\sigma \to_{\theory{R,E}} s_{n+1}\sigma
    \end{equation}
  

    Take $\rho = \theta_n\sigma$. Note that  $\rho$ satisfies $\Delta_n$ with $\Delta$:
    \begin{enumerate}[\bf (H1)]
        \item $\Delta\vdash \Delta_{n+1}\sigma$.
        \item By Definition~\ref{def:narr-C}: $\Delta_{n+1}\vdash \Delta_n\theta_n$.
        \item From {\bf (H2)} and \theory{E}-compatibility Definition~\ref{def:e-compatible}(1) generalised to \theory{E}:  $\langle\Delta_{n+1}\sigma\rangle_{nf} \vdash \Delta_n\rho$.
    \end{enumerate}
     Thus, from {\bf (H1)} and {\bf (H3)} it follows that  $\Delta\vdash \Delta_n\rho$.
    By the induction hypothesis,  we have 
    $$\Delta \vdash s_0\theta_0\ldots\theta_{n-1}\rho \rightarrow^{n}_{\theory{R,E}} s_n\rho$$
    with $\Delta\vdash \Delta_i\rho_i$ 
    and $\rho_i = \theta_i\ldots\theta_{n-1}\rho$, for every $i=1,\ldots, n$. Hence,
    $$\Delta \vdash s_0\theta_0\ldots\theta_{n-1}\theta_n\sigma \to^{n}_{\theory{R,E}} s_n\theta_n\sigma \overset{\eqref{eq:narrow}}{\to}_{\theory{R,E}} s_{n+1}\sigma,$$
    and the result follows.
%
\end{itemize}
%
\end{proof}



\begin{figure}[!t]
    \centering
    \includegraphics[scale=.55]{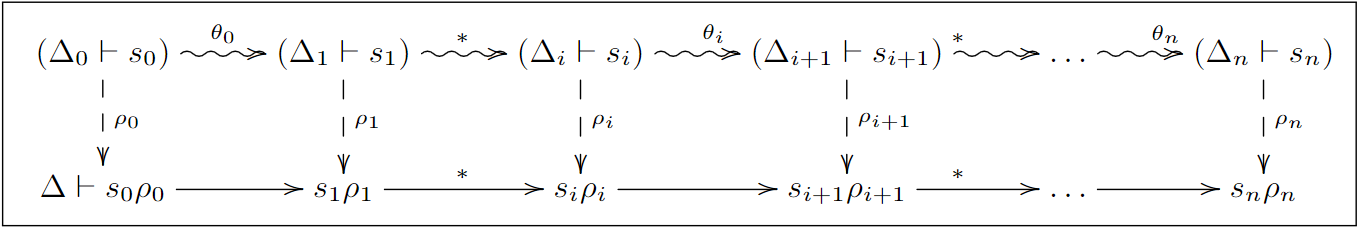}
    \caption{Corresponding Narrowing to Rewriting Derivations}
    \label{fig:narrrew}
\end{figure}

The proof of the converse (completeness) is more challenging. Nevertheless, for one-step rewriting to one-step narrowing, the result holds with no further problems:

\begin{lemma}{($\to_{\theory{R,E}}$ to $\rightsquigarrow_{\theory{R,E}}$)}\label{alem:rewrite-narr}
 Let $\Delta_0\vdash s_0$ be a nominal term in context and $V_0$ a finite set of variables containing $V=V(\Delta_0,s_0)$.
Let   $\rho_0$ be a $\theory{R,E}$-normalised substitution, with $\dom{\rho_0}\subseteq V$, that satisfies $\Delta_0$ with  $\Delta$ and
$$\Delta \vdash s_0\rho_0=t_0 \to_{\theory{R,E}} t_1.$$
Then, there exists a nominal  \theory{R,E}-narrowing step
$$(\Delta_0 \vdash s_0) \rightsquigarrow^\theta_{\theory{R,E}} (\Delta_1 \vdash s_1),$$
for a substitution $\theta$, a finite set of variables $V_1 \supseteq V(s_0)$, and a $\theory{R,E}$-normalised substitution $\rho_1$ with $\Delta$ such that

\[(i)~\Delta \vdash s_1\rho_1 \ealeq t_1 \qquad  (ii)~\texttt{dom}(\rho_1) \subseteq V_1 \qquad (iii)~\Delta\vdash \rho_0|_V\ealeq \theta \rho_1|_V \]
\end{lemma}
\begin{proof}
Suppose that the one-step rewriting is done in a position $\context{C}_0$ of $t_0$, with substitution $\sigma$ and  rule $R_0 = \nabla_0 \vdash l_0 \to r_0 \in \theory{R}$:
\begin{mathpar}
    \inferrule*[left=(*)]{t_0\equiv \context{C}_0[t_0'] \\ \Delta \vdash \nabla_0\sigma,\; t_0' \ealeq \pi\cdot(l_0\sigma)  ,\; \context{C}_0[\pi\cdot(r_0\sigma)] \aleq t_1}{\Delta \vdash t_0 \to_{[\context{C}_0, R_0],\theory{E}} t_1}
\end{mathpar}
%

The following hold:
\begin{enumerate}[{\bf (H1)}]
    \item The variables of $R_0$ are renamed with respect to $t_0 = s_0\rho_0$ and $\Delta$ (to avoid conflicts). Thus,  $V(R_0)\cap V(\Delta, t_0) = \emptyset$ and $\texttt{dom}(\sigma) \cap V_0 = \emptyset$.  
    \item By hypothesis, $\Delta\vdash \Delta_0\rho_0$ 
    \item Since $\rho_0$ is  normalised in $\Delta$ and $\Delta\vdash s_0\rho_0\to_\theory{R,E} t_1$,  there must exist a non-variable position $\context{C}_0'$ and a subterm $s_0'$ of $s_0$ such that $s_0\equiv \context{C}_0'[s_0']$ and $\Delta \vdash s_0'\rho_0 \ealeq t_0' \ealeq \pi\cdot(l_0\sigma)$.
\end{enumerate}


Define {\bf (H4)} $\theta=\rho_0\sigma$. 
Then, we have the following:
\begin{enumerate}[{\bf (H1)}]
\setcounter{enumi}{4}
    \item $\Delta \vdash \Delta_0\theta$: from {\bf (H2)} it follows that $\Delta \vdash \Delta_0\rho_0$ and $\sigma$ does not affect $\Delta_0$ since $\texttt{dom}(\sigma) = V(R_0)$.
    \item Note that $s_0'\theta=s_0'\rho_0\sigma=s_0'\rho_0$ from {\bf (H1)}. Therefore, $\Delta \vdash s_0'\theta \ealeq \pi\cdot (l_0\theta)$ and $\Delta \vdash\nabla_0\theta$,
    and  $(\Delta,\theta)$
    is a solution for the nominal \theory{E}-unification problem $(\Delta_0 \vdash s_0') \; _?\overset{\theory{E}}{{\approx}}_? \; (\nabla_0 \vdash \pi\cdot l_0)$. That is, $(\Delta, \theta)\in \mathcal{U}_\theory{E}(\Delta_0 \vdash s_0', \nabla_0 \vdash \pi\cdot l_0)$.
\end{enumerate}

Define $s_1$ as $s_1=\context{C}_0'[\pi\cdot r_0]\theta$ and $\Delta_1=\Delta$. Conditions {\bf (H4)} to {\bf (H6)} imply the existence of the  following nominal  \theory{E}-narrowing step:
 $$(\Delta_0 \vdash s_0) \rightsquigarrow^\theta_{\theory{R,E}} (\Delta_1 \vdash s_1).$$

Also $\Delta\vdash s_1=\context{C}_0'[\pi\cdot r_0]\theta=(\context{C}_0'[\pi\cdot r_0])\rho_0\sigma \aleq (\context{C}_0'\rho_0)[\pi\cdot r_0]\rho_0 \overset{(*)}{\approx}_{\alpha}  t_1$. Take $\rho_1={\tt Id}$ as the identity substitution, and items (i) and (ii) follow with $\dom{\rho_1}=\emptyset$ and $V_1=V(s_0)$.
Finally, it is trivial to check (iii):   $\Delta\vdash X\rho_0\ealeq  X\theta\rho_1\equiv X\rho_0\sigma{\tt Id}$, for all $X\in V$, and the result follows.

\end{proof}




\smallskip
\smallskip

\begin{lemma}{($\to^*_{\theory{R,E}}$ to $\rightsquigarrow^*_{\theory{R,E}}$)}\label{theo:rewritenarr}
Let $\theory{R}{\cup}\theory{E}$ be an ENRS such that $\theory{R}$ is 
\theory{E}-terminating, \theory{R,E} is \theory{E}-confluent
and $\to_{\theory{R,E}}$ is $\theory{E}$-coherent. Let $V_0$ be a finite set of variables containing $V=V(\Delta_0,s_0)$. Then, for any $\theory{R,E}$-derivation 
$$\Delta \vdash t_0 = s_0\rho_0\rightarrow_{\theory{R,E}} t_1\rightarrow_{\theory{R,E}}\ldots \rightarrow_{\theory{R,E}} t_n=t_0{\downarrow}$$
to any of its $\theory{R,E}$-normal forms, say $t_0{\downarrow}$, where $\texttt{dom}(\rho_0) \subseteq V(s_0) \subseteq V_0$ and $\rho_0$ is a $\theory{R,E}$-normalised substitution that satisfies $\Delta_0$ with $\Delta$, there exist a $\theory{R,E}$-narrowing derivation 
$$(\Delta_0\vdash s_0) \rightsquigarrow^{\theta_0}_{\theory{R,E}} (\Delta_1 \vdash s_1) \rightsquigarrow^{\theta_1}_{\theory{R,E}}  \ldots \rightsquigarrow^{\theta_{n-1}}_{\theory{R,E}} (\Delta_n \vdash s_n) $$
\begin{center}
(1) $\Delta\vdash \Delta_i\rho_i$; \qquad 
(2)  $\Delta \vdash s_i\rho_i \ealeq t_i$;\qquad 
(3)  $\Delta \vdash \rho_0|_V \ealeq \theta\rho_n|_V$.
\end{center}
where  $\rho_{i}=\theta_i\ldots \theta_{n-1}\rho$ and  $\theta=\theta_0\theta_1\ldots \theta_{n-1}$.
\end{lemma}

\begin{proof}
 By induction on the number of steps $n$ applied in the derivation $\Delta \vdash t_0 = s_0\rho_0\rightarrow_{\theory{R,E}}^* t_0{\downarrow}$.

\begin{itemize}
    \item \textbf{Base Case:} For $n=1$ the result follows directly from Lemma~\ref{alem:rewrite-narr}.

    \item \textbf{Induction Step:} Let $n>1$ and assume that the result holds for sequences of  $n-1$ rewriting steps.
    Then, 
    $$\Delta \vdash t_0 = s_0\rho_0\rightarrow_{\theory{R,E}} t_1\overbrace{\rightarrow_{\theory{R,E}}\ldots \rightarrow_{\theory{R,E}}}^{n-1} t_n=t_0{\downarrow}$$
    Now using Lemma~\ref{alem:rewrite-narr} on the rewrite step $\Delta \vdash t_0 \to_{\theory{R,E}} t_1$.  Then, we get that 
    $(\Delta_{0} \vdash s_{0}) \rightsquigarrow_{\theory{R,E}}^{\theta_{0}} (\Delta_1 \vdash s_1),$
   where $\rho_0$ is a $\theory{R,E}$-normalised substitution that satisfies $\Delta_0$ with $\Delta$, and
   \begin{enumerate}[{\bf (H1)}] 
   \item  $\Delta \vdash s_1\rho_1 = t_1' \ealeq t_1$; and 
   \item $\Delta \vdash \rho_0|_V \ealeq \theta\rho_1|_V$. 
   \end{enumerate} 
   Now consider the sequence to any of the normal forms of $t_1$:
      $$\Delta \vdash  t_1\overbrace{\rightarrow_{\theory{R,E}}\ldots \rightarrow_{\theory{R,E}}}^{n-1} t_n=t_1{\downarrow_\theory{R,E}}
      $$     
    By the induction hypothesis, there exists  a narrowing sequence 

    $$ (\Delta_1 \vdash s_1) \rightsquigarrow^{\theta_1}_{\theory{R,E}}  \ldots \rightsquigarrow^{\theta_{n-1}}_{\theory{R,E}} (\Delta_n \vdash s_n) $$
    with $\theta=\theta_1\ldots \theta_{n-1}$, a normalised substitution $\rho_n$ such that 
   \begin{enumerate}[{\bf (H1)}] 
   \setcounter{enumi}{2}
   \item $\Delta\vdash \Delta_i\rho_i$;
   \item  $\Delta \vdash s_i\rho_i =t_i'\ealeq t_i$, for every $i$; 
   \item $\Delta \vdash \rho_0|_V \ealeq \theta\rho_n|_V$. 
   \end{enumerate} 
   Note that from {\bf (H4)}, $\Delta\vdash s_n\rho_n=t_n'\ealeq t_n$. Since \theory{R} is \theory{E}-convergent and $\to_{\theory{R,E}}$ is \theory{E}-coherent, it follows from Theorem~\ref{prop:nom-e-coherence}, that all the normal forms  of $t_0$ are $\ealeq$-equivalent. That is, 
   $\Delta\vdash t_n'\ealeq t_1\downarrow_\theory{R,E}\ealeq t_0\downarrow_\theory{R,E}=t_n.$
    Therefore, there exists  a nominal \theory{E}-narrowing sequence
  $$(\Delta_0\vdash s_0) \rightsquigarrow^{\theta_0}_{\theory{R,E}} (\Delta_1 \vdash s_1) \rightsquigarrow^{\theta_1}_{\theory{R,E}}  \ldots \rightsquigarrow^{\theta_{n-1}}_{\theory{R,E}} (\Delta_n \vdash s_n).$$
\end{itemize}
\end{proof}
As a consequence of Lemmas~\ref{theo:narrowtorewrite} and~\ref{theo:rewritenarr} we obtain:
\begin{theorem}[$\theory{E}$-Lifting Theorem]\label{teo:Clifitng} 
Let $\theory{R}{\cup}\theory{E}$ be an ENRS such that \theory{E} is compatible with $\vdash$ by substitutions,  $\theory{R}$ is  \theory{E}-terminating, \theory{R,E} is \theory{E}-confluent and $\to_{\theory{R,E}}$ is $\theory{E}$-coherent.
To each finite sequence of nominal \theory{E}-rewriting steps corresponds a  finite sequence of nominal $\theory{E}$-narrowing steps, and vice versa.
\end{theorem}
Since there exists an algorithm for nominal \theory{C}-unification and \theory{C} is compatible with substitutions (Proposition~\ref{lem:nf-subst}), we have the following result.
\begin{corollary}
    The \theory{C}-Nominal Lifting theorem holds.
\end{corollary}

\section{Conclusion and Future Work}
\label{sect:future-work}

In this work, we proposed definitions for nominal $\theory{R,E}$-rewriting and $\theory{R,E}$-narrowing and proved some properties relating them, obtaining the proof of the \theory{E}-Lifting Theorem, in the case \theory{R}  is an \theory{E}-convergent NRS, $\to_\theory{R,E}$ is \theory{E}-coherent and a complete algorithm for nominal \theory{E}-unification exists. %
As $\theory{C}$ is the only equational theory for which a complete algorithm for nominal unification exists, we illustrate our results using this theory. Also, since the nominal \theory{C}-unification problem (when using freshness constraints)  only is finitary, our nominal $\theory{C}$-narrowing tree is infinitely branching. In future work, we plan to investigate alternative approaches to nominal \theory{C}-unification for which the representation of solutions is finite, such as the approach using fixed-point constraints.

\bibliographystyle{eptcs}
\bibliography{generic}

\begin{thebibliography}{10}
\providecommand{\bibitemdeclare}[2]{}
\providecommand{\surnamestart}{}
\providecommand{\surnameend}{}
\providecommand{\urlprefix}{Available at }
\providecommand{\url}[1]{\texttt{#1}}
\providecommand{\href}[2]{\texttt{#2}}
\providecommand{\urlalt}[2]{\href{#1}{#2}}
\providecommand{\doi}[1]{doi:\urlalt{https://doi.org/#1}{#1}}
\providecommand{\eprint}[1]{arXiv:\urlalt{https://arxiv.org/abs/#1}{#1}}
\providecommand{\bibinfo}[2]{#2}

\bibitemdeclare{inproceedings}{lopstr:NominalC-Unif/Washington17}
\bibitem{lopstr:NominalC-Unif/Washington17}
\bibinfo{author}{Mauricio \surnamestart Ayala{-}Rinc{\'{o}}n\surnameend},
  \bibinfo{author}{Washington \surnamestart de~Carvalho~Segundo\surnameend},
  \bibinfo{author}{Maribel \surnamestart Fern{\'{a}}ndez\surnameend} \&
  \bibinfo{author}{Daniele \surnamestart Nantes{-}Sobrinho\surnameend}
  (\bibinfo{year}{2017}): \emph{\bibinfo{title}{Nominal C-Unification}}.
\newblock In \bibinfo{editor}{Fabio \surnamestart Fioravanti\surnameend} \&
  \bibinfo{editor}{John~P. \surnamestart Gallagher\surnameend}, editors:
  {\slshape \bibinfo{booktitle}{Logic-Based Program Synthesis and
  Transformation - 27th International Symposium, {LOPSTR} 2017, Namur, Belgium,
  October 10-12, 2017, Revised Selected Papers}}, {\slshape
  \bibinfo{series}{Lecture Notes in Computer Science}} \bibinfo{volume}{10855},
  \bibinfo{publisher}{Springer}, pp. \bibinfo{pages}{235--251},
  \doi{10.1007/978-3-319-94460-9\_14}.

\bibitemdeclare{inproceedings}{OnSolvingNomFPequations/frocos/Washington17}
\bibitem{OnSolvingNomFPequations/frocos/Washington17}
\bibinfo{author}{Mauricio \surnamestart Ayala{-}Rinc{\'{o}}n\surnameend},
  \bibinfo{author}{Washington \surnamestart de~Carvalho~Segundo\surnameend},
  \bibinfo{author}{Maribel \surnamestart Fern{\'{a}}ndez\surnameend} \&
  \bibinfo{author}{Daniele \surnamestart Nantes{-}Sobrinho\surnameend}
  (\bibinfo{year}{2017}): \emph{\bibinfo{title}{On Solving Nominal Fixpoint
  Equations}}.
\newblock In \bibinfo{editor}{Clare \surnamestart Dixon\surnameend} \&
  \bibinfo{editor}{Marcelo \surnamestart Finger\surnameend}, editors: {\slshape
  \bibinfo{booktitle}{Frontiers of Combining Systems - 11th International
  Symposium, FroCoS 2017, Bras{\'{\i}}lia, Brazil, September 27-29, 2017,
  Proceedings}}, {\slshape \bibinfo{series}{Lecture Notes in Computer Science}}
  \bibinfo{volume}{10483}, \bibinfo{publisher}{Springer}, pp.
  \bibinfo{pages}{209--226},
  \doi{10.1007/978-3-319-66167-4\_12}.

\bibitemdeclare{inproceedings}{AFormalisationofNomCmatching/entcs/Ayala-RinconSFN19}
\bibitem{AFormalisationofNomCmatching/entcs/Ayala-RinconSFN19}
\bibinfo{author}{Mauricio \surnamestart Ayala{-}Rinc{\'{o}}n\surnameend},
  \bibinfo{author}{Washington \surnamestart de~Carvalho~Segundo\surnameend},
  \bibinfo{author}{Maribel \surnamestart Fern{\'{a}}ndez\surnameend} \&
  \bibinfo{author}{Daniele \surnamestart Nantes{-}Sobrinho\surnameend}
  (\bibinfo{year}{2018}): \emph{\bibinfo{title}{A Formalisation of Nominal
  {C}-Matching through Unification with Protected Variables}}.
\newblock In \bibinfo{editor}{Beniamino \surnamestart Accattoli\surnameend} \&
  \bibinfo{editor}{Carlos \surnamestart Olarte\surnameend}, editors: {\slshape
  \bibinfo{booktitle}{Proceedings of the 13th Workshop on Logical and Semantic
  Frameworks with Applications, {LSFA} 2018, Fortaleza, Brazil, September
  26-28, 2018}}, {\slshape \bibinfo{series}{Electronic Notes in Theoretical
  Computer Science}} \bibinfo{volume}{344}, \bibinfo{publisher}{Elsevier}, pp.
  \bibinfo{pages}{47--65}, \doi{10.1016/j.entcs.2019.07.004}.

\bibitemdeclare{article}{A-C-AC/tcs/Ayala-RinconSFN19}
\bibitem{A-C-AC/tcs/Ayala-RinconSFN19}
\bibinfo{author}{Mauricio \surnamestart Ayala{-}Rinc{\'{o}}n\surnameend},
  \bibinfo{author}{Washington \surnamestart de~Carvalho~Segundo\surnameend},
  \bibinfo{author}{Maribel \surnamestart Fern{\'{a}}ndez\surnameend},
  \bibinfo{author}{Daniele \surnamestart Nantes{-}Sobrinho\surnameend} \&
  \bibinfo{author}{Ana Cristina~Rocha \surnamestart Oliveira\surnameend}
  (\bibinfo{year}{2019}): \emph{\bibinfo{title}{A formalisation of nominal
  \emph{{\(\alpha\)}}-equivalence with {A}, {C}, and {AC} function symbols}}.
\newblock {\slshape \bibinfo{journal}{Theor. Comput. Sci.}}
  \bibinfo{volume}{781}, pp. \bibinfo{pages}{3--23},
  \doi{10.1016/j.tcs.2019.02.020}.

\bibitemdeclare{article}{FormalisingNomC-unif/mscs/Gabriel21}
\bibitem{FormalisingNomC-unif/mscs/Gabriel21}
\bibinfo{author}{Mauricio \surnamestart Ayala{-}Rinc{\'{o}}n\surnameend},
  \bibinfo{author}{Washington \surnamestart de~Carvalho~Segundo\surnameend},
  \bibinfo{author}{Maribel \surnamestart Fern{\'{a}}ndez\surnameend},
  \bibinfo{author}{Gabriel~Ferreira \surnamestart Silva\surnameend} \&
  \bibinfo{author}{Daniele \surnamestart Nantes{-}Sobrinho\surnameend}
  (\bibinfo{year}{2021}): \emph{\bibinfo{title}{Formalising nominal
  C-unification generalised with protected variables}}.
\newblock {\slshape \bibinfo{journal}{Math. Struct. Comput. Sci.}}
  \bibinfo{volume}{31}(\bibinfo{number}{3}), pp. \bibinfo{pages}{286--311},
  \doi{10.1017/S0960129521000050}.

\bibitemdeclare{inproceedings}{NominalNarrowing16}
\bibitem{NominalNarrowing16}
\bibinfo{author}{Mauricio \surnamestart Ayala{-}Rinc{\'{o}}n\surnameend},
  \bibinfo{author}{Maribel \surnamestart Fern{\'{a}}ndez\surnameend} \&
  \bibinfo{author}{Daniele \surnamestart Nantes{-}Sobrinho\surnameend}
  (\bibinfo{year}{2016}): \emph{\bibinfo{title}{Nominal Narrowing}}.
\newblock In \bibinfo{editor}{Delia \surnamestart Kesner\surnameend} \&
  \bibinfo{editor}{Brigitte \surnamestart Pientka\surnameend}, editors:
  {\slshape \bibinfo{booktitle}{1st International Conference on Formal
  Structures for Computation and Deduction, {FSCD} 2016, June 22-26, 2016,
  Porto, Portugal}}, {\slshape \bibinfo{series}{LIPIcs}}~\bibinfo{volume}{52},
  \bibinfo{publisher}{Schloss Dagstuhl - Leibniz-Zentrum f{\"{u}}r Informatik},
  pp. \bibinfo{pages}{11:1--11:17},
  \doi{10.4230/LIPIcs.FSCD.2016.11}.

\bibitemdeclare{inproceedings}{FPconstraints/Ayala-RinconFN18}
\bibitem{FPconstraints/Ayala-RinconFN18}
\bibinfo{author}{Mauricio \surnamestart Ayala{-}Rinc{\'{o}}n\surnameend},
  \bibinfo{author}{Maribel \surnamestart Fern{\'{a}}ndez\surnameend} \&
  \bibinfo{author}{Daniele \surnamestart Nantes{-}Sobrinho\surnameend}
  (\bibinfo{year}{2018}): \emph{\bibinfo{title}{Fixed-Point Constraints for
  Nominal Equational Unification}}.
\newblock In \bibinfo{editor}{H{\'{e}}l{\`{e}}ne \surnamestart
  Kirchner\surnameend}, editor: {\slshape \bibinfo{booktitle}{3rd International
  Conference on Formal Structures for Computation and Deduction, {FSCD} 2018,
  July 9-12, 2018, Oxford, {UK}}}, {\slshape \bibinfo{series}{LIPIcs}}
  \bibinfo{volume}{108}, \bibinfo{publisher}{Schloss Dagstuhl - Leibniz-Zentrum
  f{\"{u}}r Informatik}, pp. \bibinfo{pages}{7:1--7:16},
  \doi{10.4230/LIPIcs.FSCD.2018.7}.

\bibitemdeclare{inproceedings}{entcs:PVS/AnaCristina16}
\bibitem{entcs:PVS/AnaCristina16}
\bibinfo{author}{Mauricio \surnamestart Ayala{-}Rinc{\'{o}}n\surnameend},
  \bibinfo{author}{Maribel \surnamestart Fern{\'{a}}ndez\surnameend} \&
  \bibinfo{author}{Ana Cristina~Rocha \surnamestart Oliveira\surnameend}
  (\bibinfo{year}{2015}): \emph{\bibinfo{title}{Completeness in {PVS} of a
  Nominal Unification Algorithm}}.
\newblock In \bibinfo{editor}{Mario R.~F. \surnamestart Benevides\surnameend}
  \& \bibinfo{editor}{Ren{\'{e}} \surnamestart Thiemann\surnameend}, editors:
  {\slshape \bibinfo{booktitle}{Proceedings of the Tenth Workshop on Logical
  and Semantic Frameworks, with Applications, {LSFA} 2015, Natal, Brazil,
  August 31 - September 1, 2015}}, {\slshape \bibinfo{series}{Electronic Notes
  in Theoretical Computer Science}} \bibinfo{volume}{323},
  \bibinfo{publisher}{Elsevier}, pp. \bibinfo{pages}{57--74},
  \doi{10.1016/j.entcs.2016.06.005}.

\bibitemdeclare{inproceedings}{CICM:AyalaRinconFSKN23}
\bibitem{CICM:AyalaRinconFSKN23}
\bibinfo{author}{Mauricio \surnamestart Ayala{-}Rinc{\'{o}}n\surnameend},
  \bibinfo{author}{Maribel \surnamestart Fern{\'{a}}ndez\surnameend},
  \bibinfo{author}{Gabriel~Ferreira \surnamestart Silva\surnameend},
  \bibinfo{author}{Temur \surnamestart Kutsia\surnameend} \&
  \bibinfo{author}{Daniele \surnamestart Nantes{-}Sobrinho\surnameend}
  (\bibinfo{year}{2023}): \emph{\bibinfo{title}{Nominal AC-Matching}}.
\newblock In \bibinfo{editor}{Catherine \surnamestart Dubois\surnameend} \&
  \bibinfo{editor}{Manfred \surnamestart Kerber\surnameend}, editors: {\slshape
  \bibinfo{booktitle}{Intelligent Computer Mathematics - 16th International
  Conference, {CICM} 2023, Cambridge, UK, September 5-8, 2023, Proceedings}},
  {\slshape \bibinfo{series}{Lecture Notes in Computer Science}}
  \bibinfo{volume}{14101}, \bibinfo{publisher}{Springer}, pp.
  \bibinfo{pages}{53--68}, \doi{10.1007/978-3-031-42753-4\_4}.

\bibitemdeclare{book}{Baader98}
\bibitem{Baader98}
\bibinfo{author}{Franz \surnamestart Baader\surnameend} \&
  \bibinfo{author}{Tobias \surnamestart Nipkow\surnameend}
  (\bibinfo{year}{1998}): \emph{\bibinfo{title}{Term rewriting and all that}}.
\newblock \bibinfo{publisher}{Cambridge University Press},
  \doi{10.1017/CBO9781139172752}.

\bibitemdeclare{article}{Matching/jcss/CalvesF10}
\bibitem{Matching/jcss/CalvesF10}
\bibinfo{author}{Christophe \surnamestart Calv{\`{e}}s\surnameend} \&
  \bibinfo{author}{Maribel \surnamestart Fern{\'{a}}ndez\surnameend}
  (\bibinfo{year}{2010}): \emph{\bibinfo{title}{Matching and alpha-equivalence
  check for nominal terms}}.
\newblock {\slshape \bibinfo{journal}{J. Comput. Syst. Sci.}}
  \bibinfo{volume}{76}(\bibinfo{number}{5}), pp. \bibinfo{pages}{283--301},
  \doi{10.1016/j.jcss.2009.10.003}.

\bibitemdeclare{inproceedings}{NomSyntax/fct/DominguezF19}
\bibitem{NomSyntax/fct/DominguezF19}
\bibinfo{author}{Jes{\'{u}}s \surnamestart Dom{\'{\i}}nguez\surnameend} \&
  \bibinfo{author}{Maribel \surnamestart Fern{\'{a}}ndez\surnameend}
  (\bibinfo{year}{2019}): \emph{\bibinfo{title}{Nominal Syntax with Atom
  Substitutions: Matching, Unification, Rewriting}}.
\newblock In \bibinfo{editor}{Leszek~Antoni \surnamestart
  Gasieniec\surnameend}, \bibinfo{editor}{Jesper \surnamestart
  Jansson\surnameend} \& \bibinfo{editor}{Christos \surnamestart
  Levcopoulos\surnameend}, editors: {\slshape \bibinfo{booktitle}{Fundamentals
  of Computation Theory - 22nd International Symposium, {FCT} 2019, Copenhagen,
  Denmark, August 12-14, 2019, Proceedings}}, {\slshape
  \bibinfo{series}{Lecture Notes in Computer Science}} \bibinfo{volume}{11651},
  \bibinfo{publisher}{Springer}, pp. \bibinfo{pages}{64--79},
  \doi{10.1007/978-3-030-25027-0\_5}.

\bibitemdeclare{inproceedings}{VariantNarrowing:EscobarMS09}
\bibitem{VariantNarrowing:EscobarMS09}
\bibinfo{author}{Santiago \surnamestart Escobar\surnameend},
  \bibinfo{author}{Jos{\'{e}} \surnamestart Meseguer\surnameend} \&
  \bibinfo{author}{Ralf \surnamestart Sasse\surnameend} (\bibinfo{year}{2008}):
  \emph{\bibinfo{title}{Variant Narrowing and Equational Unification}}.
\newblock In \bibinfo{editor}{Grigore \surnamestart Rosu\surnameend}, editor:
  {\slshape \bibinfo{booktitle}{Proceedings of the Seventh International
  Workshop on Rewriting Logic and its Applications, {WRLA} 2008, Budapest,
  Hungary, March 29-30, 2008}}, {\slshape \bibinfo{series}{Electronic Notes in
  Theoretical Computer Science}} \bibinfo{volume}{238},
  \bibinfo{publisher}{Elsevier}, pp. \bibinfo{pages}{103--119},
  \doi{10.1016/j.entcs.2009.05.015}.

\bibitemdeclare{article}{NominalRewriting/FernandezG07}
\bibitem{NominalRewriting/FernandezG07}
\bibinfo{author}{Maribel \surnamestart Fern{\'{a}}ndez\surnameend} \&
  \bibinfo{author}{Murdoch \surnamestart Gabbay\surnameend}
  (\bibinfo{year}{2007}): \emph{\bibinfo{title}{Nominal rewriting}}.
\newblock {\slshape \bibinfo{journal}{Inf. Comput.}}
  \bibinfo{volume}{205}(\bibinfo{number}{6}), pp. \bibinfo{pages}{917--965},
  \doi{10.1016/j.ic.2006.12.002}.

\bibitemdeclare{article}{VarBinding/GabbayP02}
\bibitem{VarBinding/GabbayP02}
\bibinfo{author}{Murdoch \surnamestart Gabbay\surnameend} \&
  \bibinfo{author}{Andrew~M. \surnamestart Pitts\surnameend}
  (\bibinfo{year}{2002}): \emph{\bibinfo{title}{A New Approach to Abstract
  Syntax with Variable Binding}}.
\newblock {\slshape \bibinfo{journal}{Formal Aspects Comput.}}
  \bibinfo{volume}{13}(\bibinfo{number}{3-5}), pp. \bibinfo{pages}{341--363},
  \doi{10.1007/s001650200016}.

\bibitemdeclare{inproceedings}{Hullot80}
\bibitem{Hullot80}
\bibinfo{author}{Jean{-}Marie \surnamestart Hullot\surnameend}
  (\bibinfo{year}{1980}): \emph{\bibinfo{title}{Canonical Forms and
  Unification}}.
\newblock In \bibinfo{editor}{Wolfgang \surnamestart Bibel\surnameend} \&
  \bibinfo{editor}{Robert~A. \surnamestart Kowalski\surnameend}, editors:
  {\slshape \bibinfo{booktitle}{5th Conference on Automated Deduction, Les
  Arcs, France, July 8-11, 1980, Proceedings}}, {\slshape
  \bibinfo{series}{Lecture Notes in Computer Science}}~\bibinfo{volume}{87},
  \bibinfo{publisher}{Springer}, pp. \bibinfo{pages}{318--334},
  \doi{10.1007/3-540-10009-1\_25}.

\bibitemdeclare{inproceedings}{Jouannaud83:ConfluentandCoherent}
\bibitem{Jouannaud83:ConfluentandCoherent}
\bibinfo{author}{Jean{-}Pierre \surnamestart Jouannaud\surnameend}
  (\bibinfo{year}{1983}): \emph{\bibinfo{title}{Confluent and Coherent
  Equational Term Rewriting Systems: Application to Proofs in Abstract Data
  Types}}.
\newblock In \bibinfo{editor}{Giorgio \surnamestart Ausiello\surnameend} \&
  \bibinfo{editor}{Marco \surnamestart Protasi\surnameend}, editors: {\slshape
  \bibinfo{booktitle}{CAAP'83, Trees in Algebra and Programming, 8th
  Colloquium, L'Aquila, Italy, March 9-11, 1983, Proceedings}}, {\slshape
  \bibinfo{series}{Lecture Notes in Computer Science}} \bibinfo{volume}{159},
  \bibinfo{publisher}{Springer}, pp. \bibinfo{pages}{269--283},
  \doi{10.1007/3-540-12727-5\_16}.

\bibitemdeclare{inproceedings}{JouannaudKK83:Incremental}
\bibitem{JouannaudKK83:Incremental}
\bibinfo{author}{Jean{-}Pierre \surnamestart Jouannaud\surnameend},
  \bibinfo{author}{Claude \surnamestart Kirchner\surnameend} \&
  \bibinfo{author}{H{\'{e}}l{\`{e}}ne \surnamestart Kirchner\surnameend}
  (\bibinfo{year}{1983}): \emph{\bibinfo{title}{Incremental Construction of
  Unification Algorithms in Equational Theories}}.
\newblock In \bibinfo{editor}{Josep \surnamestart D{\'{\i}}az\surnameend},
  editor: {\slshape \bibinfo{booktitle}{Automata, Languages and Programming,
  10th Colloquium, Barcelona, Spain, July 18-22, 1983, Proceedings}}, {\slshape
  \bibinfo{series}{Lecture Notes in Computer Science}} \bibinfo{volume}{154},
  \bibinfo{publisher}{Springer}, pp. \bibinfo{pages}{361--373},
  \doi{10.1007/BFb0036921}.

\bibitemdeclare{inproceedings}{Kikuchi22}
\bibitem{Kikuchi22}
\bibinfo{author}{Kentaro \surnamestart Kikuchi\surnameend}
  (\bibinfo{year}{2022}): \emph{\bibinfo{title}{Ground Confluence and Strong
  Commutation Modulo Alpha-Equivalence in Nominal Rewriting}}.
\newblock In \bibinfo{editor}{Helmut \surnamestart Seidl\surnameend},
  \bibinfo{editor}{Zhiming \surnamestart Liu\surnameend} \&
  \bibinfo{editor}{Corina~S. \surnamestart Pasareanu\surnameend}, editors:
  {\slshape \bibinfo{booktitle}{Theoretical Aspects of Computing - {ICTAC} 2022
  - 19th International Colloquium, Tbilisi, Georgia, September 27-29, 2022,
  Proceedings}}, {\slshape \bibinfo{series}{Lecture Notes in Computer Science}}
  \bibinfo{volume}{13572}, \bibinfo{publisher}{Springer}, pp.
  \bibinfo{pages}{255--271},
  \doi{10.1007/978-3-031-17715-6\_17}.

\bibitemdeclare{inproceedings}{Kikuchi020}
\bibitem{Kikuchi020}
\bibinfo{author}{Kentaro \surnamestart Kikuchi\surnameend} \&
  \bibinfo{author}{Takahito \surnamestart Aoto\surnameend}
  (\bibinfo{year}{2020}): \emph{\bibinfo{title}{Confluence and Commutation for
  Nominal Rewriting Systems with Atom-Variables}}.
\newblock In \bibinfo{editor}{Maribel \surnamestart
  Fern{\'{a}}ndez\surnameend}, editor: {\slshape
  \bibinfo{booktitle}{Logic-Based Program Synthesis and Transformation - 30th
  International Symposium, {LOPSTR} 2020, Bologna, Italy, September 7-9, 2020,
  Proceedings}}, {\slshape \bibinfo{series}{Lecture Notes in Computer Science}}
  \bibinfo{volume}{12561}, \bibinfo{publisher}{Springer}, pp.
  \bibinfo{pages}{56--73},
  \doi{10.1007/978-3-030-68446-4\_3}.

\bibitemdeclare{article}{NomUnif/Schmidt-Schauss22}
\bibitem{NomUnif/Schmidt-Schauss22}
\bibinfo{author}{Manfred \surnamestart Schmidt{-}Schau{\ss}\surnameend},
  \bibinfo{author}{Temur \surnamestart Kutsia\surnameend},
  \bibinfo{author}{Jordi \surnamestart Levy\surnameend}, \bibinfo{author}{Mateu
  \surnamestart Villaret\surnameend} \& \bibinfo{author}{Yunus D.~K.
  \surnamestart Kutz\surnameend} (\bibinfo{year}{2022}):
  \emph{\bibinfo{title}{Nominal Unification and Matching of Higher Order
  Expressions with Recursive Let}}.
\newblock {\slshape \bibinfo{journal}{Fundam. Informaticae}}
  \bibinfo{volume}{185}(\bibinfo{number}{3}), pp. \bibinfo{pages}{247--283},
  \doi{10.3233/FI-222110}.

\bibitemdeclare{article}{NomUnification/UrbanPG04}
\bibitem{NomUnification/UrbanPG04}
\bibinfo{author}{Christian \surnamestart Urban\surnameend},
  \bibinfo{author}{Andrew~M. \surnamestart Pitts\surnameend} \&
  \bibinfo{author}{Murdoch \surnamestart Gabbay\surnameend}
  (\bibinfo{year}{2004}): \emph{\bibinfo{title}{Nominal unification}}.
\newblock {\slshape \bibinfo{journal}{Theor. Comput. Sci.}}
  \bibinfo{volume}{323}(\bibinfo{number}{1-3}), pp. \bibinfo{pages}{473--497},
  \doi{10.1016/j.tcs.2004.06.016}.

\bibitemdeclare{inproceedings}{Viola01}
\bibitem{Viola01}
\bibinfo{author}{Emanuele \surnamestart Viola\surnameend}
  (\bibinfo{year}{2001}): \emph{\bibinfo{title}{E-unifiability via Narrowing}}.
\newblock In \bibinfo{editor}{Antonio \surnamestart Restivo\surnameend},
  \bibinfo{editor}{Simona Ronchi~Della \surnamestart Rocca\surnameend} \&
  \bibinfo{editor}{Luca \surnamestart Roversi\surnameend}, editors: {\slshape
  \bibinfo{booktitle}{Theoretical Computer Science, 7th Italian Conference,
  {ICTCS} 2001, Torino, Italy, October 4-6, 2001, Proceedings}}, {\slshape
  \bibinfo{series}{Lecture Notes in Computer Science}} \bibinfo{volume}{2202},
  \bibinfo{publisher}{Springer}, pp. \bibinfo{pages}{426--438},
  \doi{10.1007/3-540-45446-2\_27}.

\end{thebibliography}

\end{document}